\documentclass[12pt]{article}

\usepackage{amsmath}
\usepackage{fullpage}
\usepackage{amsfonts}
\usepackage{amssymb}
\usepackage{graphicx}
\usepackage{framed}
\usepackage[top=1in,bottom=1in,left=1in,right=1in]{geometry}

\newtheorem{theorem}{Theorem}

\newtheorem{proposition}{Proposition}
\newtheorem{lemma}{Lemma}
\newtheorem{corollary}{Corollary}
\newtheorem{definition}{Definition}
\numberwithin{equation}{section}
\newenvironment{proof}[1][Proof]{\noindent\textbf{#1.} }{\ \rule{0.5em}{0.5em}}

\renewcommand{\epsilon}{\varepsilon}
\newcommand{\ket}[1]{\mathop{\left|#1\right>}\nolimits}
\newcommand{\bra}[1]{\mathop{\left<#1\,\right|}\nolimits}

\newcommand{\kb}[2]{| #1\rangle\!\langle #2 |}
\newcommand{\Tra}[1]{\mathop{{\mathrm{Tr}}_{#1}}}
\newcommand{\Tr}[2]{\mathop{{\mathrm{Tr}}_{#1}} (#2) }

\def\B{\mathcal{B}}
\def\M{\mathcal{M}}

\def\C{\mathcal{C}}
\def\D{\mathcal{D}}
\def\H{\mathcal{H}}
\def\U{\mathcal{U}}

\makeatletter
\let\@copyrightspace\relax
\makeatother

\begin{document}

\sloppy

\centerline{\Large \bf Smooth Entropy Bounds on One-Shot}
\centerline{\Large \bf Quantum State Redistribution}

\centerline{Mario Berta\footnote{berta@caltech.edu, Institute for Quantum Information and Matter, California Institute of Technology, Pasadena, California 91125, USA.}, Matthias Christandl\footnote{christandl@math.ku.dk, Department of Mathematical Sciences, University of Copenhagen, Universitetsparken 5, 2100 Copenhagen, Denmark.}, Dave Touchette\footnote{touchette.dave@gmail.com, Institute for Quantum Computing and Department of Combinatorics and Optimization, University of Waterloo, and Perimeter Institute of Theoretical Physics, and D\'epartement d'informatique et de recherche op\'erationnelle, Universit\'e de Montr\'eal, Qu\'ebec, Canada.}}


\begin{abstract}
In quantum state redistribution as introduced in [Luo and Devetak (2009)] and [Devetak and Yard (2008)], there are four systems of interest: the $A$ system held by Alice, the $B$ system held by Bob, the $C$ system that is to be transmitted from Alice to Bob, and the $R$ system that holds a purification of the state in the $ABC$ registers. We give upper and lower bounds on the amount of quantum communication and entanglement required to perform the task of quantum state redistribution in a one-shot setting. Our bounds are in terms of the smooth conditional min- and max-entropy, and the smooth max-information. The protocol for the upper bound has a clear structure, building on the work [Oppenheim (2008)]: it decomposes the quantum state redistribution task into two simpler quantum state merging tasks by introducing a coherent relay. In the independent and identical (iid) asymptotic limit our bounds for the quantum communication cost converge to the quantum conditional mutual information $I(C:R|B)$, and our bounds for the total cost converge to the conditional entropy $H(C|B)$. This yields an alternative proof of optimality of these rates for quantum state redistribution in the iid asymptotic limit. In particular, we obtain a strong converse for quantum state redistribution, which even holds when allowing for feedback.
\end{abstract}


\section{Introduction}

In the task of quantum state redistribution, we are interested in the amounts of quantum communication and entanglement that are required to transmit part of the system of one party to another party who possesses some side information about this system. It is required that all correlations, including those with any external system, are maintained. More formally, consider two parties Alice and Bob, with Alice initially holding the $A$ and $C$ registers, and Bob holding the $B$ register. The goal is then for Alice to transmit the $C$ register to Bob. If we consider a reference register $R$ holding a purification of the $ABC$ systems, then the global state on $ABCR$ is uncorrelated with any other external system, and it is sufficient to insure that correlations are maintained across these systems.

In the independent and identical (iid) asymptotic version, Alice and Bob want to perform this task on blocks of $n$ identical states, for $n$ large, and we are interested in the best asymptotic rates achievable. Luo and Devetak~\cite{LD09} proved a weak converse theorem in the iid asymptotic regime, stating that the quantum communication rate $q$ and the sum of the entanglement consumption rate $e$ (in terms of ebits) and the quantum communication rate $q$, must be at least
\begin{align}\label{eq:devetak_yard}
q\geq\frac{1}{2} I(C; R | B)\quad\mathrm{and}\quad e+q\geq H(C | B)\,.
\end{align}
Subsequently, Devetak and Yard~\cite{DY08,YD09} proved that these rates are also achievable and hence characterized the achievable rate region for iid asymptotic quantum state redistribution. In~\eqref{eq:devetak_yard}, $H(C)_{\rho}:=-\mathrm{Tr}[\rho_{C}\log\rho_{C}]$ denotes the von Neumann entropy,
\begin{align}
H(C|B)_{\rho}:=H(CB)_{\rho}-H(B)_{\rho}
\end{align}
the conditional von Neumann entropy, and
\begin{align}
I(C; R | B)_{\rho}:=H(BC)_{\rho}+H(BR)_{\rho}-H(B)_{\rho}-H(BCR)_{\rho}=H(C|B)_{\rho}-H(C|BR)_{\rho}
\end{align}
is the conditional quantum mutual information. Note that since the overall state $\rho_{ABCR}$ for quantum state redistribution is pure, we have the symmetry in $A$--$B$,
\begin{align}
I(C;R | B)_{\rho}=I(C;R | A)_{\rho}\,.
\end{align}
Later, the achievability proofs for quantum state redistribution were significantly simplified by Oppenheim~\cite{Opp08} and independently by Ye {\it et al.}~\cite{YBW08}. State redistribution can be seen as the most general bipartite noiseless coding problem, and indeed, other noiseless quantum coding primitives such as Schumacher source coding~\cite{Sch95}, quantum state merging~\cite{HOW05, HOW07} (including fully quantum Slepian-Wolf~\cite{ADHW09}), and quantum state splitting~\cite{ADHW09} can be obtained by considering the case of trivial $AB$, $A$ or $B$ system, respectively. Quantum state redistribution can also be understood as the fully quantum analogue of the tensor power input reverse Shannon theorem~\cite{BDHSW09, BCR11} with feedback to the sender and side information at the receiver~\cite{LD09, WDHW13}.

In recent years, there has been some effort on finding meaningful bounds for the one-shot version of these results (see, e.g., \cite{Ber08,Dup10,BD10,BCR11,DH11,DBWR14} and references therein). In the one-shot setting, instead of being interested in iid asymptotic rates, we are interested in the cost of achieving these tasks when only a single copy of the input state is available. Useful bounds are often stated in terms of smooth conditional entropies (see the theses of Renner~\cite{Ren05} and Tomamichel~\cite{Tom12}, as well as references therein). We are interested in finding bounds for the quantum communication cost of one-shot quantum state redistribution in terms of smooth conditional entropies.

In this paper we show that it is possible to implement quantum state redistribution for a pure quantum state $\rho^{ABCR}$ up to error $\epsilon$ with quantum communication cost at most
\begin{align}\label{eq:mainbound_intro}
q=\frac{1}{2}\Big[H_{\max}^{\epsilon}(C|B)_{\rho}-H_{\min}^{\epsilon}(C|BR)_{\rho}\Big]+O\big(\log(1/\epsilon)\big)\,,
\end{align}
and a total cost at most
\begin{align}\label{eq:mainbound_intro2}
e+q=H_{\max}^{\epsilon}(C|B)_{\rho}+O\big(\log(1/\epsilon)\big)\,,
\end{align}
where $e$ denotes the entanglement consumption in terms of ebits. Note that both the conditional min- and max-entropy terms appearing in~\eqref{eq:mainbound_intro} are smoothed, notwithstanding the fact that it is in general unknown how to simultaneously smooth marginals of overlapping quantum systems (see, e.g.,~\cite{DF13} and references therein). For the special case of iid states $(\rho^{ABCR})^{\otimes n}$ this then allows us to recover the asymptotically optimal rates as in~\eqref{eq:devetak_yard} by means of the fully quantum asymptotic equipartition property for smooth conditional entropies~\cite{TCR09}. We also speculate on how to improve the one-shot achievability bound~\eqref{eq:mainbound_intro} for the quantum communication cost with the help of embezzling entangled quantum states~\cite{vDH03} along with some of the ideas in~\cite{BCR11}.

Moreover, we give lower bounds for the quantum communication cost and the entanglement consumption in terms of the smooth conditional min- and max-entropy, and the smooth max-information. Our achievability bounds~\eqref{eq:mainbound_intro} and~\eqref{eq:mainbound_intro2} do match these lower bounds in the asymptotic iid limit. In particular, we lift the weak converse for quantum state redistribution from~\eqref{eq:devetak_yard} to a strong converse. We show that this even remains true when allowing for feedback (back communication from Bob to Alice).

Our achievability bounds originate from the authors' manuscripts~\cite{BC10,Tou14b}, where~\cite{BC10} is partly based on the preliminary results from~\cite{BRWW09}. Independently, the achievability bound~\eqref{eq:mainbound_intro} has also been derived by Datta {\it et al.}~\cite{DHO11} (stated as in preparation in~\cite{DH11}). We mention that our bound~\eqref{eq:mainbound_intro} has already proven useful for applications in quantum complexity theory. In particular, based on the idea of quantum information complexity, one of the authors~\cite{Tou14a} was able to derive the first multi-round direct sum theorem in quantum communication complexity. Our converse bounds are based on the authors' manuscripts~\cite{BC10,Tou14b}, and for Sections~\ref{sec:convent} and~\ref{sec:conventint} also on discussions with Leditzky about his work in Refs~\cite{LD15, LWD15}: 
the strong converse property for the sum of entanglement cost and quantum communication cost for state redistribution was proved using the so-called R\'enyi entropy method by Leditzky and Datta~\cite{LD15}.

This document is structured as follows. In Section~\ref{sec:prel}, we introduce our notation and give the definitions of the relevant smooth entropy measures. We then discuss in Section~\ref{sec:decoupling} the work of Oppenheim~\cite{Opp08} arguing that quantum state redistribution can be optimally decomposed into two applications of quantum state merging. Our achievability bound for one-shot quantum state redistribution is then derived in Section~\ref{sec:achievability}. The converse bounds can be found in Section~\ref{sec:converse}. We end with some conclusions (Section~\ref{sec:conclusion}).


\section{Preliminaries}\label{sec:prel}

We assume that all Hilbert spaces are finite-dimensional. The dimension of the Hilbert space associated to a quantum system $A$ is denoted by $|A|$. The set of linear operators on a quantum system $A$ is denoted by $\mathcal{L}(A)$, and the set of linear, nonnegative operators by $\mathcal{P}(A)$. We denote by $\D_\leq (A)$ the set of sub-normalized states on $A$, i.e., the set of operators $\rho^{A}\in\mathcal{P}(A)$ that are positive semi-definite, denoted $\rho \geq 0$, and have trace at most one. The set of normalized states is denoted by $\D_{=}(A)$, and $\rho_{A}\in\D_= (A)$ is a pure state if it has rank one with $\rho_{A}=\ket{\rho}\bra{\rho}_{A}$. Multipartite systems are described by the tensor product of the Hilbert spaces, and denoted by $A\otimes B$. Given a multipartite state $\rho^{AB}\in\D_{\leq}(A\otimes B)$, we write $\rho^{A}=\mathrm{Tr}_{B}[\rho^{AB}]$ for the reduced state on the system $A$. For $M^{A}\in\mathcal{L}(A)$, we write $M^{A}=M^{A}\otimes I^{B}$ for the enlargement on any $A\otimes B$, where $I^{B}$ denotes the identity in $\mathcal{P}(B)$. Quantum channels are described by completely positive trace preserving maps from some input $\mathcal{L}(A)$ to some output $\mathcal{L}(B)$.


\subsection{Entropy Measures}

The max-relative entropy of $\rho \in \D_\leq (A)$ with respect to $\sigma \in \mathcal{P}(A)$  is defined as 
\begin{align}
D_{\max} (\rho \| \sigma):=\inf \left\{ \lambda \in \mathbb{R} : 2^\lambda \sigma \geq \rho \right\}\,.
\end{align}
The conditional min-entropy of $A$ given $B$ for $\rho^{AB} \in \D_\leq (A \otimes B)$ is defined as
\begin{align}
H_{\min} (A|B)_\rho:=- \inf_{\sigma \in \D_= (B)} D_{\max} \left(\rho^{AB} \| I^A \otimes \sigma^B\right)\,.
\end{align}
The conditional max-entropy of $A$ given $B$ for $\rho^{AB} \in \D_\leq (A \otimes B)$, with purification $\rho^{ABR} \in \D_\leq (A \otimes B \otimes R)$ for some system $R$, is defined as
\begin{align}
H_{\max} (A|B)_\rho:= - H_{\min} (A|R)_\rho\,.
\end{align}
Note that this definition does not depend on the choice of the purification. The max-information that $B$ has about $A$ for $\rho^{AB} \in \D_\leq (A \otimes B)$ is defined as
\begin{align}
I_{\max} (A : B)_\rho:=  \inf_{\sigma \in \D_= (B)} D_{\max} \left(\rho^{AB} \| \rho^A \otimes \sigma^B\right)\,.
\end{align}
To define smooth entropy measures, an optimization over a set of nearby states is performed. The distance measure used is the purified distance, defined for $\rho, \sigma \in \D_\leq (A)$ as~\cite{TCR10}
\begin{align}
P (\rho, \sigma):= \sqrt{1 - \bar{F}^2 (\rho, \sigma)}\,,
\end{align}
in which the generalized fidelity is defined in terms of the fidelity as
\begin{align}
\bar{F} (\rho, \sigma):= F(\rho, \sigma) + \sqrt{(1 - \Tr{}{\rho})(1 - \Tr{}{\sigma})}\,.
\end{align}
We have $F (\rho, \sigma):= \| \sqrt{\rho} \sqrt{\sigma} \|_1$, and $\| M \|_1 = \Tr{}{\sqrt{M M^\dagger}}$. We then define an $\epsilon$-ball around $\rho \in \D_\leq (A)$ as
\begin{align}
\B^{\epsilon} (\rho):= \left\{ \bar{\rho} \in \D_\leq (A) : P (\rho, \bar{\rho}) \leq \epsilon \right\}\,.
\end{align}
For $\epsilon \geq 0$, the smooth conditional min-entropy of $A$ given $B$ for $\rho^{AB} \in \D_\leq (A \otimes B)$ is then defined as
\begin{align}
H_{\min}^\epsilon (A | B)_\rho:=  \sup_{\bar{\rho} \in \B^\epsilon (\rho^{AB})} H_{\min} (A|B)_{\bar{\rho}}\,,
\end{align}
and the smooth conditional max-entropy as
\begin{align}
H_{\max}^\epsilon (A | B)_\rho:=  \inf_{\bar{\rho}^{AB} \in \B^\epsilon (\rho^{AB})} H_{\max} (A|B)_{\bar{\rho}}\,.
\end{align}
The smooth max-information that $B$ has about $A$ for $\rho^{AB} \in \D_\leq (A \otimes B)$ is defined as
\begin{align}
I_{\max}^\epsilon (A : B)_\rho:=  \inf_{\bar{\rho}^{AB} \in \B^\epsilon (\rho^{AB})} I_{\max} (A :B)_{\bar{\rho}}\,.
\end{align}


\subsection{Quantum Asymptotic Equipartition Property}

In the asymptotic iid limit, the smooth conditional min-entropy and the conditional entropy become identical. This is captured by the fully quantum asymptotic equipartition property (AEP)~\cite{TCR09}. We make use of the following stronger version from~\cite{Tom12}.

\begin{theorem}[Asymptotic Equipartition Property~\cite{TCR09,Tom12}]\label{th:fqaep}
For any $\epsilon \in (0, 1)$ and $\epsilon^\prime \in (0, 1 - \epsilon]$, there exists $n_0 \in \mathbb{N}$ such that for any $n \geq n_0$ and $\rho^{AB} \in \D_{\leq}(A \otimes B)$ with purifying register $R$,
\begin{align}
\frac{1}{n} H_{\min}^\epsilon (A^{\otimes n} | B^{\otimes n}) & \geq H (A | B) - \frac{\delta (\epsilon, v)}{\sqrt{n}}\,,\\
\frac{1}{n} H_{\min}^\epsilon (A^{\otimes n} | B^{\otimes n}) & \leq H (A | B)  + \frac{\delta (\epsilon', v)}{\sqrt{n}} + \frac{h(\epsilon, \epsilon^\prime)}{n}\,,
\end{align}
where
\begin{align}
&\delta (\epsilon, v):= 4 \log v \sqrt{\log (2 / \epsilon^2)},\quad v:= \sqrt{2^{H_{\max} (A | B)}} + \sqrt{2^{H_{\max} (A | R)}} + 1,\quad\mathrm{and}\\
&h(\epsilon, \epsilon^\prime):= \log\left(\frac{1}{1 - \left(\epsilon\sqrt{1-\epsilon^{\prime2}} + \epsilon^\prime\sqrt{1-\epsilon^2}\right)^2}\right)\,.
\end{align}
\end{theorem}


\section{Decoupling Approach to State Redistribution}\label{sec:decoupling}

We want to make use of the following observation of Oppenheim \cite{Opp08}: quantum state redistribution can be optimally decomposed into two applications of quantum state merging by introducing a coherent relay, and applying an ebit repackaging sub-protocol. In more details, we consider four distinct parties, each holding a register. Charlie holds register $C$, that he wants to transmit to Bob, who holds register $B$, and to do so he may use help from Alice acting as a coherent relay, who holds register $A$. The state $\rho$ in registers $ABC$ is purified by state $\rho^{ABCR}$ with the $R$ register held by some reference party Ray. The goal is to transmit $C$ to Bob while minimizing the communication from Alice to Bob, and while keeping the overall correlation with Ray. We might also keep track of communication between Charlie and Alice, as well as of the entanglement consumption and generation between both Charlie and Alice, and Alice and Bob, but here our main focus is the communication between Alice and Bob. A key observation is that applying a single decoupling unitary at Charlie's side suffice to generate two hypothetical state merging protocols. Firstly, the state merging protocol that directly transmits the $C$ register to $B$ while considering both the $A$ and $R$ registers as reference. In an iid asymptotic setting, this state merging protocol requires quantum communication rate of $\frac{1}{2} I (C; AR)$ and generates ebits between Charlie and Bob at a rate $\frac{1}{2} I (C; B)$. Secondly, if we instead consider the state merging protocol that transmits the $C$ register to Alice, this requires communication of $\frac{1}{2} I (C ; RB)$ qubits between Charlie and Alice, and generates $\frac{1}{2} I (C; A)$ ebits between Charlie and Alice. As Oppenheim noted, this pure state entanglement between Alice and Charlie should not be communicated to Bob. The state redistribution protocol that uses Alice as a coherent relay then runs as follows.

Charlie merges his state with Alice's, generating $\frac{1}{2} I (C; A)$ ebits between them. Alice then replaces these ebits by some pre-shared ebits between her and Bob. This is the ebit repackaging sub-protocol, which effectively acts as a communication of $I (C; A)$ qubits between Charlie and Bob in the direct merging protocol. Alice then transmits the remaining qubits required to complete the direct merging protocol between Charlie and Bob.  A communication of
\begin{align}
\frac{1}{2} I (C; AR) - \frac{1}{2} I (C; A) =  \frac{1}{2} I(C; R | B)
\end{align}
is required to achieve this, which is asymptotically optimal~\eqref{eq:devetak_yard}. We formalize this idea below while using it in a one-shot setting and expressing the relevant bounds in terms of smooth conditional entropies.

Following the decoupling approach to quantum information theory~\cite{HHYW08, Dup10, DBWR14}, quantum state merging is conveniently understood in terms of decoupling theorems. Here we first restate the central decoupling theorem of~\cite{BCR11} in terms of smooth conditional min-entropy.

\begin{theorem}[Decoupling Theorem~\cite{BCR11}]\label{th:bcr11}
For $\epsilon > 0, \rho^{AR} \in \D_{\leq}(A \otimes R)$, and any decomposition $A = A_1 \otimes A_2$, if
\begin{align}
\log |A_1| \leq \frac{1}{2} \log |A|  + \frac{1}{2} H_{\min} (A|R)_\rho - \log \frac{1}{\epsilon}\,,
\end{align}
then
\begin{align}
\int_{\U (A)} \left\| \Tra{A_2}\big[U^{A \rightarrow A_1 A_2} \big(\rho^{A R}\big)\big] - \pi^{A_1} \otimes \rho^R \right\|_1 dU \leq \epsilon\,,
\end{align}
where dU is the Haar measure over the unitaries on system $A$, normalized to $\int dU = 1$, and $\pi^{A_1}$ is the completely mixed state on $A_1$.
\end{theorem}

For our purpose we need the following bi-decoupling result in terms of smooth conditional entropies, a direct generalization of Theorem~\ref{th:bcr11}.\footnote{A similar bi-decoupling result appears in~\cite{YBW08}, with bounds in terms of register dimensions instead of smooth conditional entropies. It would be possible to apply ideas similar to theirs to obtain a different coding theorem achieving the same achievability bound~\eqref{eq:mainbound_intro} for one-shot quantum state redistribution.}

\begin{corollary}[Bi-Decoupling Theorem]\label{cor:bidec}
For any $\epsilon_1, \epsilon_2 > 0, \rho_1^{C R_1} \in \D_{\leq} (C \otimes R_1), \rho_2^{C R_2} \in \D(C \otimes R_2)$ and any decomposition $C = C_1 \otimes C_2 \otimes C_3$, if
\begin{align}
\log |C_1| &\leq \frac{1}{2} \log |C|  + \frac{1}{2} H_{\min} (C|R_1)_{\rho_1} - \log \frac{1}{\epsilon_1}, \\
\log |C_2| &\leq \frac{1}{2} \log |C|  + \frac{1}{2} H_{\min} (C|R_2)_{\rho_2} - \log \frac{1}{\epsilon_2}\,,
\end{align}
then there exists a unitary $U^{C \rightarrow C_1 C_2 C_3}$ such that
\begin{align}
\left\| \Tra{C_2 C_3} \big[U \big(\rho_1^{C R_1}\big)\big] - \pi^{C_1} \otimes \rho_1^{R_1} \right\|_1  &\leq 3 \epsilon_1\quad\mathrm{and}\quad\left\| \Tra{C_1 C_3} \big[U \big(\rho_2^{C R_2}\big)\big] - \pi^{C_2} \otimes \rho_2^{R_2} \right\|_1  &\leq 3 \epsilon_2\,.
\end{align}
\end{corollary}

\begin{proof}
By Markov's inequality, if the condition on $|C_1|, |C_2|$ are satisfied, then Theorem \ref{th:bcr11} says that the probability over the Haar measure on $\U (C)$ that $\| \Tr{C_2 C_3} U (\rho_1)^{C R_1} - \pi^{C_1} \otimes \rho_1^{R_1} \|  \geq 3 \epsilon_1$ is at most $\frac{1}{3}$, and similarly for $\| \Tr{C_1 C_3} U (\rho_2)^{C R_2} - \pi^{C_2} \otimes \rho_2^{R_2} \|  \geq 3 \epsilon_2$, so by the union bound there is at least probability $\frac{1}{3}$ that none of these is satisfied, and then the condition of the corollary are satisfied for all corresponding $U$'s.
\end{proof}


\section{Achievability Bounds}\label{sec:achievability}

We now state formally the definition of one-shot quantum state redistribution. Let $\rho_{ABC}$ be the joint initial state of Alice and Bob's systems, where $AC$ is with Alice and $B$ is with Bob. We can view this state as part of a larger pure state $\rho^{ABCR}$ that includes a reference system $R$. In this picture quantum state redistribution means that Alice can send the $C$-part of $\rho^{ABCR}$ to Bob's side without altering the joint state. We first consider a particular setting where we have free entanglement assistance between Alice and Bob, and the goal is to minimize the number of qubits communicated from Alice to Bob in order to achieve the state transfer.

\begin{definition}[One-Shot Quantum State Redistribution]\label{def:state_redistribution}
Let $\rho^{ABC}\in\D_{\leq} (A\otimes B\otimes C)$, and let $T_{A}^{in} T_{B}^{in}$, $T_A^{out} T_B^{out}$ be additional systems. A cptp map $\Pi:ACT_{A}^{in}\otimes BT_{B}^{in}\to AT_{A}^{out}\otimes C'BT_{B}^{out}$ is called quantum state redistribution of $\rho_{ABC}$ with error $\epsilon\geq0$, if it consists of local operations and sending $q(\Pi)$ qubits with respect to the bipartition $ACT_{A}^{in}\to AT_{A}^{out}$ vs.~$BT_{B}^{in}\to C'BT_{B}^{out}$, and
\begin{align}
P\Big(\big(\Pi^{ACT_{A}^{in}BT_{B}^{in}\to AT_{A}^{out}C'BT_{B}^{out}}\big)\big(\Phi_{1}^{T_{A}^{in}T_{B}^{in}}\otimes\rho^{ABCR}\big),\Phi_{2}^{T_{A}^{out}T_{B}^{out}}\otimes\rho^{ABC'R}\Big)\leq\epsilon\,,
\end{align}
where $\rho^{ABC'R}=(I^{C\to C'}\otimes I^{ABR})\rho^{ABCR}$ for a purification $\rho^{ABCR}$ of $\rho^{ABC}$, and $\Phi_{1}$, $\Phi_{2}$ are arbitrary states on $T_{A}^{in}T_{B}^{in}$ and $T_A^{out} T_B^{out}$, respectively. The number $q$ is called quantum communication cost of the protocol $\Pi$.
\end{definition}

We obtain the following direct coding theorem for one-shot quantum state redistribution.

\begin{theorem}[Achievability One-Shot Quantum State Redistribution]\label{th:oneshqsr}
Let $\epsilon_1, \epsilon_2 \geq 0, \epsilon_3, \epsilon_4 > 0$, and $ \rho^{ABC} \in \D_{=}(A \otimes B \otimes C) $ purified by $\rho^{ABCR}$ for some register $R$. Then, there exists a quantum state redistribution $\Pi$ of $\rho_{ABC}$ with error $(8 \epsilon_1 + 2 \epsilon_2 + 4 \sqrt{3 \epsilon_3} + \sqrt{3 \epsilon_4})$ and quantum communication cost $q(\Pi)$ satisfying
\begin{align}\label{eq:achievability}
q(\Pi)\leq \frac{1}{2} H_{\max}^{\epsilon_1} (C|B)_\rho - \frac{1}{2} H_{\min}^{\epsilon_2} (C|BR)_\rho + \log \frac{1}{\epsilon_3} + \log \frac{1}{\epsilon_4} + 2\,.
\end{align}
Moreover, $\Pi$ only uses ebits as pre-shared entanglement and also generates ebit pairs. The net entanglement consumption cost $e(\Pi)$ satisfies
\begin{align}
e(\Pi) \leq \frac{1}{2} H_{\max}^{\epsilon_1} (C|B)_\rho + \frac{1}{2} H_{\min}^{\epsilon_2} (C|BR)_\rho - \log \frac{1}{\epsilon_3} + \log \frac{1}{\epsilon_4} + 1\,.
\end{align}
\end{theorem}

\begin{proof}
We first prove the theorem for the special case $\epsilon_1 = \epsilon_2 = 0$. In Corollary~\ref{cor:bidec}, we take $R_1 = BR, R_2 = AR, \rho_1 = \rho^{CBR}, \rho_2 = \rho^{CAR}$,
\begin{align}
\log |C_1| &= \left\lfloor \frac{1}{2} \log |C|  + \frac{1}{2} H_{\min} (C|BR)_{\rho} - \log \frac{1}{\epsilon_3} \right\rfloor\\
\log |C_2| &= \left\lfloor \frac{1}{2} \log |C|  + \frac{1}{2} H_{\min} (C|AR)_{\rho} - \log \frac{1}{\epsilon_4} \right\rfloor\,,
\end{align}
and then there exists a unitary $U^{C \rightarrow C_1 C_2 C_3}$ satisfying
\begin{align}
\left\| \Tra{C_2 C_3} \big[U\big(\rho^{CBR}\big)\big] - \pi^{C_1} \otimes \rho^{BR} \right\|_1 \leq 3 \epsilon_3\quad\mathrm{and}\quad\left\| \Tra{C_1 C_3} \big[U\big(\rho^{CAR}\big)\big] - \pi^{C_2} \otimes \rho^{AR} \right\|_1 \leq 3 \epsilon_4\,.
\end{align}
We transform these in  purified distance bounds using a generalized Fuchs-van der Graaf inequality,
\begin{align}
\label{eq:ucbr}
P \Big( \Tra{C_2 C_3} \big[U \big(\rho^{CBR}\big)\big], \pi^{C_1} \otimes \rho^{BR} \Big) &\leq \sqrt{3 \epsilon_3},\\
\label{eq:ucar}
P \Big( \Tra{C_1 C_3}\big[U \big(\rho^{CAR}\big)\big], \pi^{C_2} \otimes \rho^{AR} \Big) &\leq \sqrt{3 \epsilon_4}\,.
\end{align}
Let $A^\prime, A^{\prime \prime}$ be isomorphic to $A$, $B^{\prime \prime \prime}$ be isomorphic to $B$, $C^\prime, C^{\prime \prime \prime}$ be isomorphic to $C$, and $C_2^{ \prime \prime}, C_3^{\prime \prime}$ be isomorphic to $C_2, C_3$, respectively. Then, Uhlmann's theorem tells us that there exist partial isometries
\begin{align}
V_1^{C_2 C_3 A \rightarrow A_1 A^\prime C^\prime}\quad\mathrm{and}\quad V_2^{C_1 C_3 B \rightarrow B_2 B^{\prime \prime \prime} C^{\prime \prime \prime}}
\end{align}
satisfying
\begin{align}
\label{eq:ulluv1}
P\Big( V_1 U \big(\rho^{ABCR}\big), \kb{\phi_1}{\phi_1}^{A_1 C_1} \otimes I^{AC \rightarrow A^\prime C^\prime} \big(\rho^{ABCR}\big)\Big)
&= P\Big( \Tra{C_2 C_3}\big[U\big(\rho^{CBR}\big)\big],\pi^{C_1} \otimes \rho^{BR}\Big)\\
\label{eq:ulluv2}
P \Big( V_2 U \big(\rho^{ABCR}\big), \kb{\phi_2}{\phi_2}^{B_2 C_2} \otimes I^{BC \rightarrow B^{\prime \prime \prime} C^{\prime \prime \prime}} \big(\rho^{ABCR}\big) \Big)
&= P \Big( \Tra{C_1 C_3}\big[ U \big(\rho^{CAR}\big)\big], \pi^{C_2} \otimes \rho^{AR} \Big)\,.
\end{align}
Let $T_A, T_B$ be isomorphic to $A_1, C_1$, respectively, and denote by
\begin{align}
\hat{U}^{ C \rightarrow T_B C_2^{\prime \prime} C_3^{\prime \prime} },\quad
\hat{V}_1^{ C_2^{ \prime \prime} C_3^{\prime \prime} A^{\prime \prime} \rightarrow T_A A^\prime C^\prime}\quad\mathrm{and}\quad\hat{V}_2^{ T_B C_3^{\prime \prime} B \rightarrow  B_2 B^{\prime \prime \prime} C^{\prime \prime \prime}}
\end{align}
a version of $U$ that maps register $C$ into registers $T_B C_2^{\prime \prime} C_3^{\prime \prime}$, a version of $V_1$ that maps registers $C_2^{ \prime \prime} C_3^{\prime \prime} A^{\prime \prime}$ into registers $ T_A A^\prime C^\prime$, and a version of $V_2$ that maps registers $T_B C_3^{\prime \prime} B$ into registers $B_2 B^{\prime \prime \prime} C^{\prime \prime \prime}$, respectively. Also let $\M^{T_A A^\prime C^\prime}$ be a channel performing a projective measurement onto the image of $\hat{V}_1$, and mapping everything outside this image to some fixed state $\phi_\M^{T_A A^\prime C^\prime}$ in it. We can now define our one-shot state redistribution protocol $\Pi$ (also see Figure \ref{fig:prot}).

\begin{figure}
\centering
\includegraphics[width=14cm]{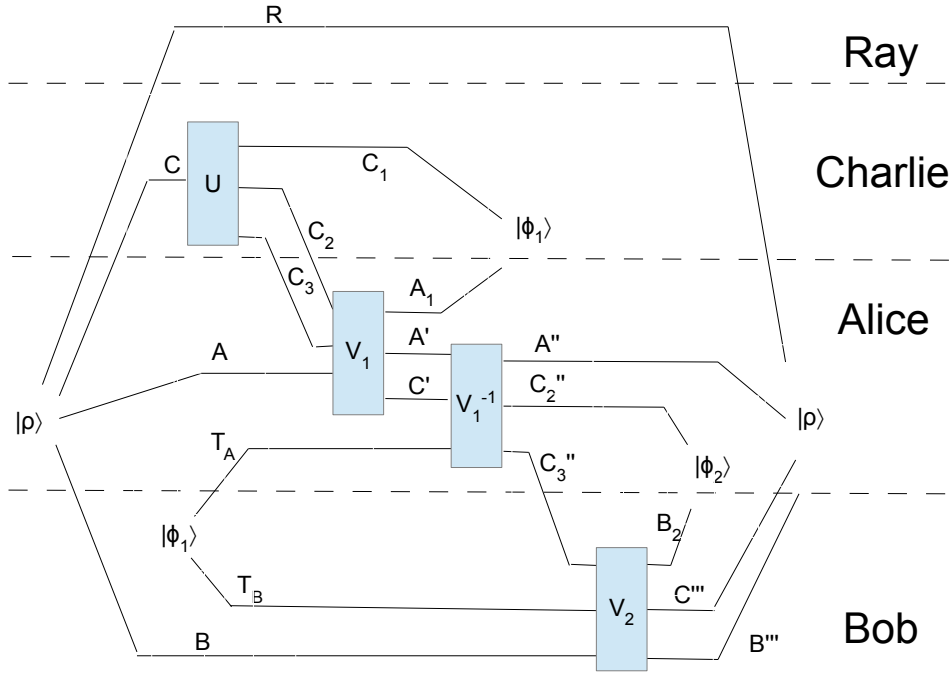}
\caption{One-shot protocol for quantum state redistribution from the ebit repackaging sub-protocol. There are four distinct parties Alice, Bob, Charlie, and Ray each holding their register $ABCR$ of the overall pure state $\rho_{ABCR}$, respectively. The goal is to transmit $C$ to Bob while minimizing the communication from Alice to Bob and keeping the overall correlation with Ray. The decoupling unitary $U$ at Charlie's side generates two hypothetical state merging protocols: one that directly transmits the $C$ register to Bob (with decoding isometry $V_{2}$, while considering the $A$ and $R$ registers as the reference, and one that transmits the $C$ register to Alice while considering the $B$ and $R$ as the reference (with decoding isometry $V_{1}$). The state redistribution protocol that uses Alice as a coherent relay then runs as follows. Charlie first merges his state with Alice's, generating ebits between them. Alice then replaces these ebits by some pre-shared ebits between her and Bob: the ebit repackaging sub-protocol. Finally, Alice transmits the remaining qubits required to complete the direct merging protocol between Charlie and Bob.}
\label{fig:prot}
\end{figure}

\newpage

\begin{framed}
\textbf{Protocol $\Pi$ for input $\rho^{ABCR}$ using ebits $\phi_1^{T_A T_B}$}
\begin{enumerate}
\item Charlie applies U on register $C$, keeps register $C_1$, and transmits the $C_2, C_3$ registers to Alice.
\item Alice applies $V_1$ on $C_2 C_3 A$, 
obtains registers $A_1 A^\prime C^\prime$, 
and then uses $T_A$ instead of $A_1$;
she performs $\M$ on $T_A A^\prime C^\prime$ to apply $\hat{V}_1^\dagger$ and 
obtains registers $ A^{\prime \prime} C_2^{\prime \prime} C_3^{\prime \prime}$.
\item Alice transmits the $C_3^{\prime \prime}$ register to Bob.
\item Bob applies $\hat{V}_2$ on $ T_B C_3^{\prime \prime} B$ and obtains registers $ B_2 B^{\prime \prime \prime} C^{\prime \prime \prime}$.
\end{enumerate}
\hrulefill
\begin{itemize}
\item The $B^{\prime \prime \prime}, C^{\prime \prime \prime}$ output registers held by Bob correspond to the $B, C$ input registers, respectively, while the $A^{\prime \prime}$ output register held by Alice corresponds to the $A$ input register. Together with the untouched reference register $R$, these should be close to $\rho^{ABCR}$.
\item The $A_1 C_1$ registers should be close to the maximally entangled state $\ket{\phi_1}^{A_1 C_1} = I^{T_A T_B \rightarrow A_1 C_1} \ket{\phi_1}^{T_A T_B}$ shared between Alice and Charlie, while the $C_2^{\prime \prime} B_2$ registers should be close to the maximally entangled state $\ket{\phi_2}^{C_2^{\prime \prime} B_2}$ shared between Alice and Bob, with Alice holding the $C_2^{\prime \prime}$ share.
\end{itemize}
\end{framed}

Note that Charlie only communicates with Alice, and the only register effectively transmitted between Alice and Bob is the $C_3^{\prime \prime}$ register, which is of the same size as the $C_3$ register. We then have the following bound on the communication,
\begin{align}
q = \log |C_3|& = \log |C| - \log |C_2| - \log |C_1| \\
& \leq - \frac{1}{2} H_{\min} (C|AR)_{\rho} - \frac{1}{2} H_{\min} (C| BR)_{\rho} + \log \frac{1}{\epsilon_4} + \log \frac{1}{\epsilon_3} + 2 \\
& = \frac{1}{2} H_{\max} (C|B)_{\rho} - \frac{1}{2} H_{\min} (C| BR)_{\rho} + \log \frac{1}{\epsilon_3} + \log \frac{1}{\epsilon_4} + 2\,.
\end{align}
Note that this protocol is based on ebits, i.e., the only pre-shared entanglement it uses are ebits. In the notation of Definition~\ref{def:state_redistribution}, $T_A^{in} = T_A, T_B^{in} = T_B, T_A^{out} = C_2^{\prime \prime}, T_B^{out} = B_2$.
The consumption and generation of ebits can also be easily computed from the above dimensions.
The consumption is $\log |T_A| = \log |C_1| \leq \frac{1}{2} \log |C|  + \frac{1}{2} H_{\min} (C|BR)_{\rho} - \log \frac{1}{\epsilon_3}$ ebits, and the number of ebits generated is $\log |C_2^{\prime \prime}| = \log |C_2| \geq \frac{1}{2} \log |C|  + \frac{1}{2} H_{\min} (C|AR)_{\rho} - \log \frac{1}{\epsilon_4} - 1$. The net entanglement cost $e$ is then bounded by
\begin{align}
e = & \log |C_1| - \log |C_2| \\
\leq & \frac{1}{2} H_{\min} (C|BR)_{\rho} - \log \frac{1}{\epsilon_3} - \frac{1}{2} H_{\min} (C|AR)_{\rho} + \log \frac{1}{\epsilon_4} + 1 \\
= & \frac{1}{2} H_{\max} (C|B)_{\rho} + \frac{1}{2} H_{\min} (C|BR)_{\rho} - \log \frac{1}{\epsilon_3} + \log \frac{1}{\epsilon_4} + 1\,.
\end{align}
It is left to verify that the final state is close enough to $\rho^{ABCR}$. We
prove a stronger result, that the global final state is close to $\rho^{ABCR} \otimes \phi_1^{A_1 C_1} \otimes \phi_2^{C_2^{\prime \prime} B_2}$. This is the criteria normally used in state-redistribution based on ebits. Properties of the purified distance that we use are proved in~\cite{TCR10, Tom12}. We first use the triangle inequality to obtain the following four terms,
\begin{align}
&P\Big(\hat{V}_2 \hat{V}_1^{-1} \M V_1U \big(\rho^{ABCR} \otimes \phi_1^{T_A T_B} \big),I^{ABC \rightarrow A^{\prime \prime} B^{\prime \prime \prime} C^{\prime \prime \prime}} \big(\rho^{ABCR} \otimes \phi_1^{A_1 C_1} \otimes \phi_2^{C_2^{\prime \prime} B_2}\big)\Big)\notag\\
&\leq P\Big(  \hat{V}_2 \hat{V}_1^{-1} \M V_1  U\big(\rho^{ABCR} \otimes \phi_1^{T_A T_B} \big),\hat{V}_2 \hat{V}_1^{-1} \M I^{AC \rightarrow A^{\prime} C^{\prime}}\big(\rho^{ABCR} \otimes \phi_1^{A_1 C_1} \otimes \phi_1^{T_A T_B} \big)\Big)\notag\\
&\quad+P\Big(  \hat{V}_2 \hat{V}_1^{-1} \M I^{AC \rightarrow A^{\prime} C^{\prime}} \big(\rho^{ABCR} \otimes \phi_1^{A_1 C_1} \otimes \phi_1^{T_A T_B} \big),\notag\\
&\quad\quad\quad\hat{V}_2 \hat{V}_1^{-1} I^{AC \rightarrow A^{\prime} C^{\prime}}\big(\rho^{ABCR} \otimes \phi_1^{A_1 C_1} \otimes \phi_1^{T_A T_B} \big)\Big)\notag\\
&\quad+P\Big( \hat{V}_2 \hat{V}_1^{-1} I^{AC \rightarrow A^{\prime} C^{\prime}} \big(\rho^{ABCR} \otimes \phi_1^{A_1 C_1} \otimes \phi_1^{T_A T_B}\big) ,\notag\\
&\quad\quad\quad\hat{V}_2 SWAP_{C_1 \leftrightarrow T_B} I^{A C_2 C_3 T_A \rightarrow A^{\prime \prime} C_2^{\prime \prime} C_3^{\prime \prime} A_1} U \big(\rho^{ABCR} \otimes \phi_1^{T_A T_B}\big)\Big)\notag\\
&\quad+P\Big( \hat{V}_2  SWAP_{C_1 \leftrightarrow T_B} I^{A C_2 C_3 T_A \rightarrow A^{\prime \prime} C_2^{\prime \prime} C_3^{\prime \prime} A_1} U \big(\rho^{ABCR} \otimes \phi_1^{T_A T_B}\big),\notag\\
&\quad\quad\quad I^{ABC \rightarrow A^{\prime \prime} B^{\prime \prime \prime} C^{\prime \prime \prime} } \big(\rho^{ABCR} \otimes \phi_1^{A_1 C_1} \otimes \phi_2^{C_2^{\prime \prime} B_2}\big)\Big)\,.
\end{align}
To bound the first term, we have
\begin{align}
&P\Big( \hat{V}_2 \hat{V}_1^{-1} \M V_1 U \big(\rho^{ABCR} \otimes \phi_1^{T_A T_B} \big),
\hat{V}_2 \hat{V}_1^{-1} \M I^{AC \rightarrow A^{\prime} C^{\prime}}\big(\rho^{ABCR} \otimes \phi_1^{A_1 C_1} \otimes \phi_1^{T_A T_B} \big)\Big)\notag\\
& \leq P\Big( V_1  U \big(\rho^{ABCR} \otimes \phi_1^{T_A T_B} \big),
I^{AC \rightarrow A^{\prime} C^{\prime}} \big(\rho^{ABCR} \otimes \phi_1^{A_1 C_1} \otimes \phi_1^{T_A T_B} \big)\Big) \\
& = P\Big( V_1  U \big(\rho^{ABCR} \big),I^{AC \rightarrow A^{\prime} C^{\prime}} \big(\rho^{ABCR} \otimes \phi_1^{A_1 C_1} \big)\Big) \\
& \leq \sqrt{3 \epsilon_3}\,.
\end{align}
The first inequality is by monotonicity of the purified distance, the first equality is because appending an uncorrelated system does not change the distance, and finally the last inequality is by combining~\eqref{eq:ucbr} and~\eqref{eq:ulluv1}. For the second term, we have
\begin{align}
& P\Big(  \hat{V}_2 \hat{V}_1^{-1} \M I^{AC \rightarrow A^{\prime} C^{\prime}} \big(\rho^{ABCR} \otimes \phi_1^{A_1 C_1} \otimes \phi_1^{T_A T_B} \big),\hat{V}_2 \hat{V}_1^{-1} I^{AC \rightarrow A^{\prime} C^{\prime}}\big(\rho^{ABCR} \otimes \phi_1^{A_1 C_1} \otimes \phi_1^{T_A T_B} \big)\Big)\notag\\
& \leq P\Big(   \M I^{AC \rightarrow A^{\prime} C^{\prime}} \big(\rho^{ABCR} \otimes \phi_1^{A_1 C_1} \otimes \phi_1^{T_A T_B} \big), I^{AC \rightarrow A^{\prime} C^{\prime}}\big(\rho^{ABCR} \otimes \phi_1^{A_1 C_1} \otimes \phi_1^{T_A T_B} \big)\Big)\\
& \leq P\Big(   \M I^{AC \rightarrow A^{\prime} C^{\prime}} \big(\rho^{ABCR}  \otimes \phi_1^{T_A T_B} \big), \M \hat{V}_1 \hat{U} I^{A \rightarrow A^{\prime \prime} }\big(\rho^{ABCR}  \big)\Big)\notag\\
& \quad+ P\Big(   \M \hat{V}_1 \hat{U} I^{A \rightarrow A^{\prime \prime} } \big(\rho^{ABCR}  \big), I^{AC \rightarrow A^{\prime} C^{\prime}}\big(\rho^{ABCR} \otimes \phi_1^{T_A T_B} \big)\Big)\\
& \leq P\Big(   I^{AC \rightarrow A^{\prime} C^{\prime}} \big(\rho^{ABCR}  \otimes \phi_1^{T_A T_B} \big),  \hat{V}_1 \hat{U} I^{A \rightarrow A^{\prime \prime} }\big(\rho^{ABCR} \big)\Big)\notag\\
& \quad+ P\Big(   \hat{V}_1 \hat{U} I^{A \rightarrow A^{\prime \prime} } \big(\rho^{ABCR}  \big), I^{AC \rightarrow A^{\prime} C^{\prime}}\big(\rho^{ABCR} \otimes \phi_1^{T_A T_B} \big)\Big)\\
& \leq 2 \sqrt{3 \epsilon_3}\,.
\end{align}
The first inequality is by monotonicity of the purified distance, the second by the triangle inequality and because appending an uncorrelated system does not change the distance, the third by monotonicity and because $\M \hat{V}_1 = \hat{V}_1$, and finally the last inequality is by combining (\ref{eq:ucbr}) and (\ref{eq:ulluv1}) twice after relabeling systems. For the third term, we have
\begin{align}
&P\Big(\hat{V}_2 \hat{V}_1^{-1}  I^{AC \rightarrow A^{\prime} C^{\prime}} \big(\rho^{ABCR} \otimes \phi_1^{A_1 C_1} \otimes \phi_1^{T_A T_B} \big), \notag\\
&\quad\hat{V}_2 SWAP_{ C_1 \leftrightarrow T_B} I^{A C_2 C_3 T_A \rightarrow A^{\prime \prime } C_2^{\prime \prime} C_3^{\prime \prime} A_1 } U \big(\rho^{ABCR} \otimes \phi_1^{T_A T_B} \big) \Big) \notag\\
&= P\Big( SWAP_{A_1 C_1 \leftrightarrow T_A T_B} I^{AC \rightarrow A^{\prime} C^{\prime}} \big(\rho^{ABCR} \otimes \phi_1^{A_1 C_1} \otimes \phi_1^{T_A T_B} \big),\notag\\
&\quad\quad SWAP_{A_1 C_1 \leftrightarrow T_A T_B} \hat{V}_1 SWAP_{C_1 \leftrightarrow T_B}
 I^{A C_2 C_3 T_A \rightarrow A^{\prime \prime} C_2^{\prime \prime} C_3^{\prime \prime} A_1 } U \big(\rho^{ABCR} \otimes \phi_1^{T_A T_B} \big)\Big) \\
&= P\Big(  I^{AC \rightarrow A^{\prime} C^{\prime}} \big(\rho^{ABCR} \otimes \phi_1^{A_1 C_1} \otimes \phi_1^{T_A T_B} \big),
V_1 U \big(\rho^{ABCR} \otimes \phi_1^{T_A T_B} \big)\Big) \\
& = P\Big(  I^{AC \rightarrow A^{\prime} C^{\prime}} \big(\rho^{ABCR} \otimes \phi_1^{A_1 C_1}  \big),V_1 U \big(\rho^{ABCR} \big)\Big) \\
& \leq \sqrt{3 \epsilon_3}\,.
\end{align}
The first equality is by isometric invariance, the second is because $SWAP_{A_1 C_1 \leftrightarrow T_A T_B}$ leaves the first state invariant and also because
\begin{align}
&SWAP_{A_1 C_1 \leftrightarrow T_A T_B} \hat{V}_1^{ C_2^{\prime \prime} C_3^{\prime \prime} A^{\prime \prime} \rightarrow  T_A A^\prime C^\prime} SWAP_{C_1 \leftrightarrow T_B} I^{A C_2 C_3 T_A  \rightarrow A^{\prime \prime} C_2^{\prime \prime} C_3^{\prime \prime} A_1 }\notag\\
&= I^{T_A T_B C_1} V_1^{ C_2 C_3 A \rightarrow A_1 A^\prime C^\prime}
\end{align}
the next is because appending uncorrelated systems does not change the distance, and finally the last inequality is by combining~\eqref{eq:ucbr} and~\eqref{eq:ulluv1}. For the fourth term, we have
\begin{align}
&P\Big( \hat{V}_2  SWAP_{C_1 \leftrightarrow T_B} I^{A C_2 C_3 T_A \rightarrow A^{\prime \prime} C_2^{\prime \prime} C_3^{\prime \prime} A_1 } U \big(\rho^{ABCR} \otimes \phi_1^{T_A T_B} \big),\notag\\
&\quad I^{ABC \rightarrow A^{\prime \prime} B^{\prime \prime \prime} C^{\prime \prime \prime} } \big(\rho^{ABCR} \otimes \phi_1^{A_1 C_1} \otimes \phi_2^{C_2^{\prime \prime} B_2}\big)\Big)\notag\\
& = P\Big( I^{A^{\prime \prime} C_2^{\prime \prime} \rightarrow A C_2} \hat{V}_2  SWAP_{C_1 \leftrightarrow T_B} I^{A C_2 C_3 T_A \rightarrow A^{\prime \prime} C_2^{\prime \prime} C_3^{\prime \prime} A_1 } U \big(\rho^{ABCR} \otimes \phi_1^{T_A T_B} \big),\notag\\
&\quad\quad I^{A^{\prime \prime} C_2^{\prime \prime} \rightarrow A C_2} I^{ABC \rightarrow A^{\prime \prime} B^{\prime \prime \prime} C^{\prime \prime \prime} } \big(\rho^{ABCR} \otimes \phi_1^{A_1 C_1} \otimes \phi_2^{C_2^{\prime \prime} B_2} \big)\Big) \\
& = P\Big( I^{T_A T_B \rightarrow A_1 C_1} V_2 U \big(\rho^{ABCR} \otimes \phi_1^{T_A T_B} \big),I^{BC \rightarrow B^{\prime \prime \prime} C^{\prime \prime \prime} } \big(\rho^{ABCR} \otimes \phi_1^{A_1 C_1} \otimes \phi_2^{C_2 B_2} \big)\Big) \\
& = P\Big(  V_2 U \big(\rho^{ABCR} \big),I^{BC \rightarrow B^{\prime \prime \prime} C^{\prime \prime \prime} } \big(\rho^{ABCR}  \otimes \phi_2^{C_2 B_2} \big)\Big) \\
& \leq \sqrt{3 \epsilon_4}\,.
\end{align}
The first equality is just a system relabeling, the second is because
\begin{align}
&I^{A^{\prime \prime} C_2^{\prime \prime} \rightarrow A C_2} \hat{V}_2^{ T_B C_3^{\prime \prime} B \rightarrow  B_2 B^{\prime \prime \prime} C^{\prime \prime \prime}}  SWAP_{C_1 \leftrightarrow T_B} I^{A C_2 C_3 T_A \rightarrow A^{\prime \prime} C_2^{\prime \prime} C_3^{\prime \prime} A_1 }\notag\\
&= I^{A C_2} I^{T_A T_B \rightarrow A_1 C_1} V_2^{ C_1 C_3 B \rightarrow  B_2 B^{\prime \prime \prime} C^{\prime \prime \prime}}
\end{align}
the third is because appending an uncorrelated system does not change the distance and finally the last inequality is by combining~\eqref{eq:ucar}) and~\eqref{eq:ulluv2}. Putting these four bounds together, we get the stated bound for $\epsilon_1, \epsilon_2 =0$, and this completes the proof for this case.

We can now prove the smooth entropy version of the theorem by extending the above argument to the states achieving the extremum in the smooth entropies. Let $\omega_1^{ABCR} \in S_\leq (\H_{ABCR})$ be such that $P (\omega_1, \rho) \leq \epsilon_1$ and $H_{\min}^{\epsilon_1} (C|BR)_\rho = H_{\min} (C|BR)_{\omega_1}$. Similarly, let $\omega_2^{ABCR} \in S_\leq (\H_{ABCR})$ be such that $P (\omega_2, \rho) \leq \epsilon_2$ and $H_{\max}^{\epsilon_2} (C|B)_\rho = H_{\max} (C|B)_{\omega_2}$, and consider a purification $\omega_2^{ABCR S_2}$. In Corollary \ref{cor:bidec}, we take $R_1 = BR, R_2 = ARS_2, \rho_1 = \omega_1^{CBR}, \rho_2 = \omega_2^{CARS_2}$,
\begin{align}
\log |C_1| &= \left\lfloor \frac{1}{2} \log |C|  + \frac{1}{2} H_{\min} (C|BR)_{\omega_1} - \log \frac{1}{\epsilon_3} \right\rfloor\\
\log |C_2| &= \left\lfloor \frac{1}{2} \log |C|  + \frac{1}{2} H_{\min} (C|AR S_2)_{\omega_2} - \log \frac{1}{\epsilon_4} \right\rfloor\,,
\end{align}
and then there exists a unitary $U^{C \rightarrow C_1 C_2 C_3}$ satisfying
\begin{align}
\left\| \Tra{C_2 C_3} \big[U \big(\omega_1^{CBR}\big)\big] - \pi^{C_1} \otimes \omega_1^{BR} \right\|_1 &\leq 3 \epsilon_3\\
\left\| \Tra{C_1 C_3} \big[U \big(\omega_2^{CARS_2}\big)\big] - \pi^{C_2} \otimes \omega_2^{ARS_2} \right\|_1 &\leq 3 \epsilon_4\,.
\end{align}
Transforming these in  purified distance bounds, we get
\begin{align}
P \Big( \Tra{C_2 C_3} \big[U \big(\omega_1^{CBR}\big)\big], \pi^{C_1} \otimes \omega_1^{BR} \Big) &\leq \sqrt{3 \epsilon_3}\\
P \Big( \Tra{C_1 C_3}\big[U \big(\omega_2^{CAR}\big)\big], \pi^{C_2} \otimes \omega_2^{AR} \Big) &\leq \sqrt{3 \epsilon_4}\,,
\end{align}
in which we also used monotonicity of the purified distance under partial trace of $S_2$.
Since
\begin{align}
P(\omega_1^{ABCR}, \rho^{ABCR}) \leq \epsilon_1\quad\mathrm{and}\quad P(\omega_2^{ABCR}, \rho^{ABCR}) \leq \epsilon_2\,,
\end{align}
the triangle inequality along with monotonicity of the purified distance and the fact that appending uncorrelated systems does not increase distance imply the bounds
\begin{align}
\label{eq:ucbromega}
P\Big( \Tra{C_2 C_3} \big[U \big(\rho^{CBR}\big)\big], \pi^{C_1} \otimes \rho^{BR} \Big) &\leq \sqrt{3 \epsilon_3} + 2 \epsilon_1\\
\label{eq:ucaromega}
P\Big( \Tra{C_1 C_3} \big[U \big(\rho^{CAR}\big)\big], \pi^{C_2} \otimes \rho^{AR} \Big) &\leq \sqrt{3 \epsilon_4} + 2 \epsilon_2\,.
\end{align}
Considering systems $A^\prime, A^{\prime \prime}, B^{\prime \prime \prime}, C^\prime, C^{\prime \prime \prime}, C_2^{ \prime \prime}, C_3^{\prime \prime}, T_A, T_B$ as above, Uhlmann's theorem tells us that there exist partial isometries
\begin{align}
V_1^{C_2 C_3 A \rightarrow A_1 A^\prime C^\prime}\quad\mathrm{and}\quad V_2^{C_1 C_3 B \rightarrow B_2 B^{\prime \prime \prime} C^{\prime \prime \prime}}
\end{align}
satisfying
\begin{align}
P\Big( V_1 U \big(\rho^{ABCR}\big), \kb{\phi_1}{\phi_1}^{A_1 C_1} \otimes I^{AC \rightarrow A^\prime C^\prime} \big(\rho^{ABCR}\big)\Big)
&= P \Big( \Tra{C_2 C_3}\big[U \big(\rho^{CBR}\big)\big], \pi^{C_1} \otimes \rho^{BR} \Big) \\
P \Big( V_2 U \big(\rho^{ABCR}\big), \kb{\phi_2}{\phi_2}^{B_2 C_2} \otimes I^{BC \rightarrow B^{\prime \prime \prime} C^{\prime \prime \prime}} \big(\rho^{ABCR}\big) \Big)
&= P \Big( \Tra{C_1 C_3}\big[U \big(\rho^{CAR}\big)\big], \pi^{C_2} \otimes \rho^{AR}\Big)\,.
\end{align}
Also consider
\begin{align}
\hat{U}^{ C \rightarrow T_B C_2^{\prime \prime} C_3^{\prime \prime} },\quad
\hat{V}_1^{ C_2^{ \prime \prime} C_3^{\prime \prime} A^{\prime \prime} \rightarrow T_A A^\prime C^\prime}\quad\mathrm{and}\quad\hat{V}_2^{ T_B C_3^{\prime \prime} B \rightarrow  B_2 B^{\prime \prime \prime} C^{\prime \prime \prime}}\,,
\end{align}
the versions of $U, V_1, V_2$ acting on the corresponding registers, as in the $\epsilon_1, \epsilon_2 = 0$ case, as well as the channel $\M^{T_A A^\prime C^\prime}$ performing a projective measurement onto the image of $\hat{V}_1$ as above. We can then take the smooth version of our one-shot state redistribution protocol~$\Pi$ to be formally defined as the non-smooth version above, but using these $U, V_1, \M, \hat{V}_1, \hat{V}_2$ instead. We then have the following bound on the communication,
\begin{align}
q = \log |C_3| &= \log |C| - \log |C_2| - \log |C_1| \\
& \leq - \frac{1}{2} H_{\min} (C|A R S_2)_{\omega_2} - \frac{1}{2} H_{\min} (C| BR)_{\omega_1} + \log \frac{1}{\epsilon_4} + \log \frac{1}{\epsilon_3} +2 \\
& = \frac{1}{2} H_{\max} (C|B)_{\omega_2} - \frac{1}{2} H_{\min} (C| BR)_{\omega_1} + \log \frac{1}{\epsilon_3} + \log \frac{1}{\epsilon_4} + 2 \\
& = \frac{1}{2} H_{\max}^{\epsilon_2} (C|B)_{\rho} - \frac{1}{2} H_{\min}^{\epsilon_1} (C| BR)_{\rho} + \log \frac{1}{\epsilon_3} + \log \frac{1}{\epsilon_4} + 2\,.
\end{align}
Similarly, we have the following bound on the net entanglement cost $e$,
\begin{align}
e = & \log |C_1| - \log |C_2| \\
\leq  & \frac{1}{2} H_{\max}^{\epsilon_2} (C|B)_{\rho} + \frac{1}{2} H_{\min}^{\epsilon_1} (C|BR)_{\rho} - \log \frac{1}{\epsilon_3} + \log \frac{1}{\epsilon_4} + 1\,.
\end{align}
Is left to verify that the final state is close enough to $\rho^{ABCR}$. The analysis is the same as in the $\epsilon_1, \epsilon_2 = 0$ case, with the bounds (\ref{eq:ucbr}) and (\ref{eq:ucar}) replaced by (\ref{eq:ucbromega}) and (\ref{eq:ucaromega}), yielding the desired bound
\begin{align}
&P\Big(\hat{V}_2 \hat{V}^{-1} \M V_1 U \big(\rho^{ABCR} \otimes \phi_1^{T_A T_B}\big),
I^{ABC \rightarrow A^{\prime \prime} B^{\prime \prime \prime} C^{\prime \prime \prime}} \big(\rho^{ABCR} \otimes \phi_1^{A_1 C_1} \otimes \phi_2^{C_2^{\prime \prime} B_2}\big)\Big)\notag\\
&\leq 8 \epsilon_1 + 2 \epsilon_2 + 4 \sqrt{3 \epsilon_3} + \sqrt{3 \epsilon_4}\,.
\end{align}
\end{proof}

Note however that, at least if we allow arbitrary shared entanglement, the above bound for the quantum communication cost can not be tight in general. This can be seen by considering the situation where the $B$ register is trivial, which corresponds to state splitting. Here it is known~\cite{BCR11} that we can succeed with quantum communication $I_{\max}^\epsilon (C; R)_\rho$ using entanglement embezzling states. This can be much smaller than the bound we provide for some states $\rho_ABCR$. We provide an alternate protocol, using entanglement embezzling states rather than standard maximally entangled states, which achieves a communication rate that is upper bounded by the smooth max-information, up to small additive terms, in the case that either the $A$ or the $B$ register is trivial. Hence, this protocol has the optimal communication cost for the special cases of state merging and state splitting.

The idea for the protocol with embezzling states is borrowed from~\cite{BCR11}, and is the following. At the outset of the protocol, before applying the above protocol as a sub-protocol, we first perform a coherent projective measurement in the eigenbasis of the $C$ system, and discard the portion with eigenvalues smaller than $|C|^2$. We then coherently apply the above ebit-based protocol on each branch, with the state in branch $i$ denoted $\rho_i$, using an entanglement embezzling state between Charlie and Alice, and another between Alice and Bob, to provide the necessary ebits, as well as to absorb any ebits created, up to small error. Different amount of ebits are generated and consumed on each branch, hence the need for entanglement embezzling states. We also transmit the register containing the coherent measurement outcomes, to allow to undo these. This procedure then flattens the eigenvalue spectra on the $C$ system, hence the min- and max-entropies $H_{\min}^{\epsilon} (C)_{\rho_i}, H_{\max}^{\epsilon} (C)_{\rho_i}$ are both equal to the rank of $\rho_i^C$, up to a small error. This allows us to replace the max-entropy term by a min entropy term when the $B$ register is trivial, and similarly when $A$ is trivial, and in such a case we can use the lemmas given in~\cite{BCR11} to relate this to smooth max-information, and obtain a provably optimal rate. See~\cite{BCR11} for a formal definition of the $\rho_i$'s. In general, the communication grows as
\begin{align}\label{eq:qrst_idea}
\frac{1}{2} \max_i \Big[H_{\max}^{\epsilon} (C|B)_{\rho_i} -  H_{\min}^{\epsilon} (C|BR)_{\rho_i}\Big]
\end{align}
up to small additive terms. This is however not optimal in general, and it is still unclear whether this can be of any help for obtaining tight bounds for state redistribution (cf.~Section~\ref{sec:conclusion}). An approach that might hold some promise could be to allow for interaction in the state redistribution protocol. For example, in a two-message protocol in which Bob speaks first, this would then allow Bob to also do some preprocessing similar to what Alice does here, and possibly allow for improved flattening in the general case.\footnote{Using the pre-processing from~\cite{BCR11} would only amount to a sub-linear communication cost from Bob to Alice, and thus vanishing back communication cost in the iid asymptotic setting.}

In the asymptotic iid regime, our bound can be used to prove that a communication at the conditional mutual information is sufficient to achieve exponentially small error.

\begin{theorem}[Exponentially Small Error Quantum State Redistribution~\cite{DY08, YD09}]\label{lem:asymptqsr}
For $\rho \in \D_= (A \otimes  B \otimes  C)$  purified by $\rho^{ABCR}$ for some register $R$ and $\mu > 0$, there exists $c > 0, n_0 \in \mathbb{N}$ such that for all $n \geq n_0$, there exists a quantum state redistribution $\Pi_n $ of $(\rho^{ABC})^{\otimes n}$ with error $2^{- c n}$ and quantum communication cost  $q (\Pi_n)$ satisfying
\begin{align}
q (\Pi_n) \leq n\cdot\left[\frac{1}{2}I(C;R|B)+\mu\right]\,.
\end{align}
\end{theorem}

\begin{proof}
This follows by applying the fully quantum asymptotic equipartition property as stated in Theorem~\ref{th:fqaep} to the bound in Theorem~\ref{th:oneshqsr}. Fix some pure state $\rho^{ABCR}$ and some $\mu > 0$, then we want to achieve error $\epsilon = 2^{-cn}$ for large enough $n$ and some $c$ that depends on both $\rho$ and $\mu$. Take $\epsilon_1 = \epsilon_2 = \sqrt{\epsilon_3} = \sqrt{\epsilon_4} = \frac{\epsilon}{20}$. For the $H_{\max}^{\epsilon_1}$ term, $v$ in the AEP is fixed once $\rho$ is fixed, so that $\delta (\epsilon, v) =  4 \log v \sqrt{2cn(1 + O (1/n))}$ and this can be made smaller than $\frac{\mu}{4}$ by taking $c$ proportional to $(\frac{\mu}{\log v})^2$. Similarly for the $H_{\min}^{\epsilon_2}$ term. Then, the $\log \frac{1}{\epsilon_3}$ and $\log \frac{1}{\epsilon_4}$ terms grow proportionally to $cn$, so that by taking $c$ also  smaller than $\mu$, the conditions of the theorem are satisfied.
\end{proof}


\section{Converse Bounds}\label{sec:converse}

\subsection{Quantum Communication}

Here we provide lower bounds on the amount of quantum communication required for one-shot state redistribution. In the iid asymptotic regime these bounds simplify to the conditional mutual information $I(C; R | B)_\rho$.

\begin{proposition}[Converse One-Shot Quantum Communication]\label{prop:converse}
Let $\epsilon_1 \in (0, 1)$ and $\epsilon_2 \in (0, 1 - \epsilon_1)$ and $\rho \in \D_= (A \otimes  B \otimes  C)$ with purification register $R$. Then, the quantum communication cost $q(\Pi)$ of every quantum state redistribution $\Pi$ of $\rho_{ABC}$ with error $\epsilon_1$ is lower bounded by
\begin{align}
q(\Pi) &\geq \frac{1}{2} I_{\max}^{\epsilon_1 + \epsilon_2} (R; BC)_\rho - \frac{1}{2} I_{\max}^{\epsilon_2} (R; B)_\rho\\
q(\Pi) &\geq \frac{1}{2} H_{\min}^{\epsilon_2 } (R| B)_\rho - \frac{1}{2} H_{\min}^{\epsilon_1 + \epsilon_2} (R| BC)_\rho\label{eq:converse_min}\\
q(\Pi) &\geq \frac{1}{2} H_{\max}^{\epsilon_1 + \epsilon_2 } (R| B)_\rho - \frac{1}{2} H_{\max}^{ \epsilon_2} (R| BC)_\rho\,,
\end{align}
and the same bounds hold for $B$ replaced with $A$.
\end{proposition}

Note that the first bound is optimal in the case of a trivial $B$ register, for state splitting, while the corresponding bound with $A$ replacing $B$ is optimal in the case of a trivial $A$ register, for state merging. Also note that, in contrast to the direct coding bound, the time-reversal symmetry between the $A, B$ systems is not apparent here. Finally, note that these bound hold irrespective of any entanglement assistance.\\

\begin{proof}[Proof of Proposition~\ref{prop:converse}]
Similar to the proof of the optimal bound on state splitting in \cite{BCR11}, we look at the correlations between Bob and Ray. To be able to use~\cite[Lemma B.9]{BCR11}, we look at the max-information that Bob has about Ray at the end of any protocol for quantum state redistribution. A one-message protocol for state redistribution necessarily has the following structure: local operation on Alice's side, followed by communication from Alice to Bob, and then local operations on Bob's side. In more details.

\newpage

\begin{framed}
\textbf{General protocol $\Pi$ for input $\rho^{ABCR}$ using entanglement $\phi^{T_A^{in} T_B^{in}} $}
\begin{enumerate}
\item Alice holds the $A,C, T_A^{in}$ systems at the outset, and Bob the $B, T_B^{in}$ systems.
\item Alice applies a local operation on the $A C T_A^{in}$ registers.  Her registers are then $T_A^{out} A^{\prime} Q$. The joint state is $\sigma^{T_A^{out} A^{\prime} Q B T_B^{in} R}$.
\item Alice transmits the $Q$ register to Bob.
\item Bob applies a local operation on $Q B T_B^{in}$. His registers are then $T_B^{out} B^{\prime} C^{\prime}$. The joint state is $\theta^{T_A^{out} T_B^{out} A^{\prime}  B^{\prime} C^{\prime} R}$.
\end{enumerate}
\hrulefill
\begin{itemize}
\item The requirement is that the $A^{\prime}  B^{\prime} C^{\prime} R$ part is $\epsilon_1$ close to $\rho^{A^\prime B^\prime C^\prime R} = I^{ABC \rightarrow A^{\prime} B^{\prime} C^{\prime}} (\rho^{ABCR})$ in purified distance.
\end{itemize}
\end{framed}

For the bound in terms of max-information, consider a state $\hat{\theta}^{A^\prime B^\prime C^\prime R} \in D_\leq (A^\prime \otimes B^\prime \otimes C^\prime \otimes R)$ such that $P(\theta^{A^\prime B^\prime C^\prime R}, \hat{\theta}^{A^\prime B^\prime C^\prime R}) \leq \epsilon_2$ and $I_{\max}^{\epsilon_2} (R; B^{\prime} C^{\prime} )_{\theta} =  I_{\max} (R; B^{\prime} C^{\prime} )_{\hat{\theta}}$. Such a state must exist by the definition of smoothing and the properties of the purified distance. Then $P (\rho^{A^\prime B^\prime C^\prime R}, \hat{\theta}^{A^\prime B^\prime C^\prime R}) \leq \epsilon_1 + \epsilon_2$ by the triangle inequality since $\theta^{A^\prime B^\prime C^\prime R}$ must be $\epsilon_1$ close to $\rho^{A^\prime B^\prime C^\prime R}$.
We get the following chain of inequalities
\begin{align}
I_{\max}^{\epsilon_1 + \epsilon_2} (R; BC)_\rho & \leq I_{\max} (R; B^{\prime} C^{\prime} )_{\hat{\theta}} \\
& = I_{\max}^{\epsilon_2} (R; B^{\prime} C^{\prime} )_{\theta} \\
& \leq I_{\max}^{\epsilon_2} (R; Q B T_B^{in} )_{\sigma} \\
& \leq I_{\max}^{\epsilon_2} (R; B T_B^{in} )_{\sigma} + 2 \log |Q| \\
& = I_{\max}^{\epsilon_2} (R; B T_B^{in} )_{\rho \otimes \phi} + 2 \log |Q| \\
& = I_{\max}^{\epsilon_2} (R; B )_{\rho} + 2 \log |Q|\,,
\end{align}
in which the first inequality follows by definition of smooth max-information and monotonicity of purified distance, since $\theta^{A^\prime B^\prime C^\prime R}$ is within distance $\epsilon$ of $\rho^{A^\prime B^\prime C^\prime R}$, the first equality is by the choice of $\hat{\theta}$, the second inequality is because the max-information is monotone under local operations, the third inequality follows by Lemma B.9 of \cite{BCR11}, the second equality is because local operations of Alice do not change the max-information of Bob about the reference, and the last is because $\phi^{T_A^{in} T_B^{in}}$ is uncorrelated to $\rho^{ABCR}$.

For the bound in terms of conditional min-entropy, we similarly get, by taking an appropriate $\hat{\theta}$ and using an unlockability property of the conditional min-entropy (Lemma~\ref{lem:unlock}, Appendix),
\begin{align}
H_{\min}^{\epsilon_1 + \epsilon_2} (R| BC)_\rho\geq H_{\min} (R | B^{\prime} C^{\prime} )_{\hat{\theta}} &= H_{\min}^{\epsilon_2} (R; B^{\prime} C^{\prime} )_{\theta} \\
& \geq H_{\min}^{\epsilon_2} (R | Q B T_B^{in} )_{\sigma} \\
& \geq H_{\min}^{\epsilon_2} (R | B T_B^{in} )_{\sigma} - 2 \log |Q| \\
& = H_{\min}^{\epsilon_2} (R | B T_B^{in} )_{\rho \otimes \phi} - 2 \log |Q| \\
& = H_{\min}^{\epsilon_2} (R | B )_{\rho} - 2 \log |Q|\,.
\end{align}
For the bound in term of the conditional max-entropy, we obtain the bound with the $A$ system instead of $B$ by using the duality relation of conditional min- and max-entropy. We then get the remaining bounds by interchanging the $A$ and $B$ systems in those already proved, and by using the symmetry of state redistribution under time reversal.
\end{proof}

In the asymptotic iid regime, the above bounds together with the fully quantum asymptotic equipartition property (Theorem~\ref{th:fqaep}) imply a strong converse for the quantum communication rate of quantum state redistribution.

\begin{theorem}[Strong Converse Quantum Communication]
For all $\rho \in \D_= (A \otimes  B \otimes  C)$  purified by $\rho^{ABCR}$ for some register $R$ and $\mu > 0$, there exists $c > 0, n_0 \in \mathbb{N}$ such that for all $n \geq n_0$ and all 
quantum state redistribution $\Pi_n $ of $(\rho^{ABC})^{\otimes n}$ with error $1 - 2^{- c n}$, the quantum communication cost  $q (\Pi_n)$ of $\Pi_n$ satisfies
\begin{align}
q (\Pi_n) \geq n\cdot\left[\frac{1}{2}I(C;R|B)-\mu\right]\,.
\end{align}
\end{theorem}

\begin{proof}
Fix some pure state $\rho^{ABCR}$ and $\mu > 0$. By applying the fully quantum asymptotic equipartition property as stated in Theorem~\ref{th:fqaep} to the converse bound~\eqref{eq:converse_min} we get for any quantum state redistribution $\Pi_n $ of $(\rho^{ABC})^{\otimes n}$ with error $\epsilon_1>0$ and $n$ large enough,
\begin{align}
\frac{1}{n}q (\Pi_n)\geq\frac{1}{2}I(C ; R | B)-\frac{\delta(\epsilon_2,v)+\delta(1-\epsilon_1-\epsilon_2,v')}{\sqrt{n}}-\frac{h(\epsilon_1+\epsilon_2,1-\epsilon_1-\epsilon_2)}{n}\,,
\end{align}
where $\delta(\epsilon_2,v)$ comes from the $R|B$ term, $\delta(1-\epsilon_1-\epsilon_2,v')$ from the $R|BC$ term, and $\epsilon_2>0$ sufficiently small. By choosing $\epsilon_2=2^{-dn}$ and $\epsilon_1=1 - 2^{-cn}$ for $c,d>0$ (depending on $\rho$ and $\mu$) such that
\begin{align}
\mu\geq\frac{\delta(\epsilon_2,v)+\delta(1-\epsilon_1-\epsilon_2,v')}{\sqrt{n}}+\frac{h(\epsilon_1+\epsilon_2,1-\epsilon_1-\epsilon_2)}{n}\,
\end{align}
the claim follows.
\end{proof}


\subsection{Interactive Communication}\label{sec:int_comm}

Even though our achievability bounds only require a single message from Alice to Bob, we show in the next section that a strong converse also hold when allowing for feedback communication from Bob to Alice. We note that a reader not interested in the feedback case can safely ignore this subsection.

In order to show this result, we use the following notation for interactive communication. In the interactive model, a $M$-message protocol $\Pi$ for a given task from input registers $A_{in}, B_{in}$ to output registers $A_{out}, B_{out}$ is defined by a sequence of isometries $U_1, \cdots, U_{M + 1}$ along with a pure state $\psi \in \D (T_A T_B)$ shared between Alice and Bob, for arbitrary finite dimensional registers $T_A, T_B$: the pre-shared entanglement. We need $M+1$ isometries in order to have $M$ messages since a first isometry is applied before the first message is sent and a last one after the final message is received. In the case of even $M$, for appropriate finite dimensional quantum memory registers $A_1, A_3, \cdots A_{M - 1}, A^\prime$ held by Alice, $B_2, B_4, \cdots B_{M - 2}, B^\prime$ held by Bob, and quantum communication registers $C_1, C_2, C_3, \cdots C_M$ exchanged by Alice and Bob, we have
\begin{align}
U_1 \in \U(A_{in} T_A, A_1 C_1),\;&U_2 \in \U(B_{in} T_B C_1, B_2 C_2),\;U_3 \in \U(A_1 C_2, A_3 C_3),\;U_4 \in \U(B_2 C_3, B_4 C_4),\notag\\
\;\cdots,\;&U_{M}\in \U(B_{M - 2} C_{M - 1},B_{out} B^\prime C_{M}),\;U_{M + 1} \in \U(A_{M - 1} C_M, A_{out} A^\prime)\,,
\end{align}
see the long version of~\cite[Figure 1]{Tou14a}. We adopt the convention that, at the outset, $A_0 = A_{in} T_A, B_0 = B_{in} T_B$, 
for odd $i$ with $1 \leq i < M$ $B_i = B_{i-1}$, for even $i$ with $1 < i \leq M$ $A_i = A_{i-1}$ and also $B_M = B_{M + 1} = B_{out} B^\prime$, and $A_{M+1} = A_{out} A^\prime$. In this way, after application of $U_i$, Alice holds register $A_i$, Bob holds register $B_i$ and the communication register is $C_i$. In the case of an odd number of message $M$, the registers corresponding to $U_M, U_{M+1}$ are changed accordingly. We slightly abuse notation and also write $\Pi$ to denote the channel in $\C (A_{in} B_{in}, A_{out} B_{out})$ implemented by the protocol, i.e., for any $\rho \in \D (A_{in} B_{in})$,
\begin{align}
\Pi (\rho):=\Tr{A^\prime B^\prime }{U_{M + 1} U_M \cdots U_2 U_1 (\rho \otimes \psi)}\,.
\end{align}
Note that the $ A^\prime $ and $ B^\prime $ registers are the final memory registers that are being discarded at the end of the protocol by Alice and Bob, respectively. We define the quantum communication cost of $\Pi$ from Alice to Bob as 
\begin{align}
QCC_{A \rightarrow B} (\Pi):=\sum_i \log C_{2i +1}\,,
\end{align}
and the quantum communication cost of $\Pi$ from Bob to Alice as
\begin{align}
QCC_{B \rightarrow A} (\Pi):=\sum_i \log C_{2i}\,.
\end{align}
The total communication cost of the protocol is then the sum of these two quantities. We have the following definition for quantum state redistribution in this interactive setting.

\begin{definition}[Feedback Quantum State Redistribution]\label{def:feed_state_redistribution}
Let $\rho^{ABC}\in\D_{\leq} (A\otimes B\otimes C)$, and let $T_{A}^{in} T_{B}^{in}$, $T_A^{out} T_B^{out}$ be additional systems. An $M$-message protocol $\Pi:ACT_{A}^{in}\otimes BT_{B}^{in}\to AT_{A}^{out}\otimes C'BT_{B}^{out}$ in the interactive model is called an $M$-message quantum state redistribution of $\rho_{ABC}$ with error $\epsilon\geq0$ if
\begin{align}
P\Big(\big(\Pi^{ACT_{A}^{in}BT_{B}^{in}\to AT_{A}^{out}C'BT_{B}^{out}}\big)\big(\Phi_{1}^{T_{A}^{in}T_{B}^{in}}\otimes\rho^{ABCR}\big),\Phi_{2}^{T_{A}^{out}T_{B}^{out}}\otimes\rho^{ABC'R}\Big)\leq\epsilon\,,
\end{align}
where $\rho^{ABC'R}=(I^{C\to C'}\otimes I^{ABR})\rho^{ABCR}$ for a purification $\rho^{ABCR}$ of $\rho^{ABC}$, and $\Phi_{1}$, $\Phi_{2}$ are arbitrary states on $T_{A}^{in}T_{B}^{in}$ and $T_A^{out} T_B^{out}$, respectively.
\end{definition}


\subsection{Free Back-Communication}\label{sec:back-comm}

Note that  asymptotic quantum state redistribution composes perfectly. That is, given any decomposition $C = D_1 D_2 \cdots D_d$, the total asymptotic cost for transmitting $C$ in a single message versus transmitting it in $d$ successive messages from Alice to Bob is the same
\begin{align}
I (C ; R | B) = I (D_1 ; R | B) + I (D_2 ; R | B D_1) + \cdots +  I (D_d ; R | B D_1 \cdots D_{d-1})\,.
\end{align}
This follows from the chain rule for conditional quantum mutual information. This considers the setting when only Alice sends messages to Bob and we consider such a decomposition of the $C$ system; by allowing feedback (back-communication) and an arbitrary protocol to transmit $C$, we could hope to improve on this. This is not possible: we show that even if there is free back-communication from Bob to Alice between Alice's multiple messages, this cannot decrease the total asymptotic cost of communication from Alice to Bob. Quantum state redistribution with feedback is defined by allowing arbitrary protocols in the  interactive model of communication, as defined in Section~\ref{sec:int_comm}. We then account only for the quantum communication cost from Alice to Bob, denoted $QCC_{A \rightarrow B}$.

\begin{proposition}[Converse One-Shot Quantum Communication with Feedback]\label{prop:multiconverse}
Let $\epsilon_1 \in (0, 1)$ and $\epsilon_2 \in (0, 1 - \epsilon_1)$ and $\rho \in \D_= (A \otimes  B \otimes  C)$ with purification register $R$. Then, for every $M$-message quantum state redistribution $\Pi$ of $\rho$ with error $\epsilon_1$, the quantum communication cost $QCC_{A \rightarrow B} (\Pi)$ from Alice to Bob of $\Pi$ is lower bounded by
\begin{align}
QCC_{A \rightarrow B} (\Pi) \geq \frac{1}{2} I_{\max}^{\epsilon_1 + \epsilon_2} (R; BC)_\rho - \frac{1}{2} I_{\max}^{\epsilon_2} (R; B)_\rho,\\
QCC_{A \rightarrow B} (\Pi) \geq \frac{1}{2} H_{\min}^{\epsilon_2 } (R| B)_\rho - \frac{1}{2} H_{\min}^{\epsilon_1 + \epsilon_2} (R| BC)_\rho,\\
QCC_{A \rightarrow B} (\Pi) \geq \frac{1}{2} H_{\max}^{\epsilon_1 + \epsilon_2 } (R| B)_\rho - \frac{1}{2} H_{\max}^{ \epsilon_2} (R| BC)_\rho\,,
\end{align}
and the same bounds hold for $B$ replaced with $A$.
\end{proposition}

Note that there is no dependence on the number $M$ of messages in these lower bounds. Hence, the strong converse for quantum state redistribution proved from the corresponding bounds in the preceding section also hold if we allow for feedback.

\begin{theorem}[Strong Converse Quantum Communication with Feedback]
For all $\rho \in \D_= (A \otimes  B \otimes  C)$  purified by $\rho^{ABCR}$ for some register $R$ and $\mu > 0$, there exists $c > 0, n_0 \in \mathbb{N}$ such that for all $n \geq n_0$ and all $M$, every $M$-message
quantum state redistribution $\Pi_n^M $ of $(\rho^{ABC})^{\otimes n}$ with error $1 - 2^{- c n}$, the quantum communication cost from Alice to Bob $QCC_{A \rightarrow B} (\Pi_n)$ of $\Pi_n$ satisfies
\begin{align}
QCC_{A \rightarrow B} (\Pi_n) \geq n\cdot\left[\frac{1}{2}I(C;R|B)-\mu\right]\,.
\end{align}
\end{theorem}

\begin{proof}[Proof of Proposition~\ref{prop:multiconverse}]
The proof is  similar to the one in the single round case. Hence, we only write down the details for the bound in terms of max-information. We consider a $M$-message protocol with local isometric processing and arbitrary pre-shared entanglement $\phi^{T_A T_B}$, for the case of even $M$ (the case of odd $M$ follows similarly). The $A_i$ registers are Alice's quantum memory registers for round $i$, for odd $i$ and after application of $U_i$, similarly for $B_i$ on Bob's side and even $i$, and the $C_i$ registers are the communication registers exchanged by Alice and Bob in round $i$, from Alice to Bob for odd $i$, and from Bob to Alice for even $i$. We use the following notation,
\begin{align}
\rho_0 := \rho^{ABCR} \otimes \phi^{T_A T_B},\;\rho_1 := U_1 (\rho_0),\;\rho_2 := U_2 (\rho_1),\;\cdots,\;\rho_{M+1} := U_{M+1} (\rho_M)\,.
\end{align}
It must hold that $P (\rho_{M+1}^{ABCR}, \rho^{ABCR}) \leq \epsilon_1$. Consider a state $\hat{\theta}^{A B C R} \in D_\leq (A  B  C  R)$ such that $P(\rho_{M+1}^{A B C R}, \hat{\theta}^{A B C R}) \leq \epsilon_2$ and $I_{\max}^{\epsilon_2} (R; B C )_{\rho_{M+1}} =  I_{\max} (R; B C )_{\hat{\theta}}$. Such a state must exist by the definition of smoothing and the properties of the purified distance. Then $P (\rho^{A B C R}, \hat{\theta}^{A B C R}) \leq \epsilon_1 + \epsilon_2$ by the triangle inequality. We get the following chain of inequalities,
\begin{align}
I_{\max}^{\epsilon_1 + \epsilon_2} (R; BC)_\rho & \leq I_{\max} (R; B C )_{\hat{\theta}} \\
& = I_{\max}^{\epsilon_2} (R; B C )_{\rho_{M+1}} \\
& = I_{\max}^{\epsilon_2} (R; B C )_{\rho_{M}} \\
& \leq I_{\max}^{\epsilon_2} (R; C_{M-1} B_{M-2} )_{\rho_{M-1}} \\
& \leq I_{\max}^{\epsilon_2} (R; B_{M-2} )_{\rho_{M-1}} + 2 \log |C_{M-1}| \\
& = I_{\max}^{\epsilon_2} (R; B_{M-2} )_{\rho_{M-2}} + 2 \log |C_{M-1}| \\
& \leq I_{\max}^{\epsilon_2} (R; C_{M-3} B_{M-4} )_{\rho_{M-3}} + 2 \log |C_{M-1}| \\
& \leq \cdots \\
& \leq I_{\max}^{\epsilon_2} (R; C_{1} B_{0} )_{\rho_{1}} + 2 \sum_{i \geq 1} \log |C_{2i + 1}| \\
& \leq I_{\max}^{\epsilon_2} (R;  B_{0} )_{\rho_{1}} + 2 \sum_{i \geq 0} \log |C_{2i + 1}| \\
& = I_{\max}^{\epsilon_2} (R;  B T_B )_{\rho_{0}} + 2 QCC_{A \rightarrow B} (\Pi) \\
& = I_{\max}^{\epsilon_2} (R;  B )_{\rho} + 2 QCC_{A \rightarrow B} (\Pi)\,.
\end{align}
The first inequality follows by definition of smooth max-information and monotonicity of purified distance. The first equality is by the choice of $\hat{\theta}$, and the second because $U_{M+1}$ is applied on Alice's side. The second inequality is because the max-information is monotone under local operations. The third inequality follows by Lemma~\ref{lem:b9}, and the third equality is because local operations of Alice do not change the max-information of Bob about the reference. The following sequence of inequality follows by applying the last few ones repeatedly. The last inequality follows by Lemma~\ref{lem:b9}, the following equality is by definition of $B_0 = B T_B$ and $ QCC_{A \rightarrow B} (\Pi) = \sum_{i \geq 0} \log |C_{2i + 1}|$,
and the last is because $\phi^{T_A^{in} T_B^{in}}$ is uncorrelated to $\rho^{ABCR}$.
\end{proof}


\subsection{Entanglement Consumption}\label{sec:convent}

The above results give strong converses on the amount of quantum communication required from Alice to Bob, even if we allow for interaction between Alice and Bob. We would now like to obtain bounds on the total amount of resources, including net entanglement consumption and back-communication, required in an interactive protocol. For this we restrict the allowed entanglement assistance in the definition of state redistribution (Definition~\ref{def:state_redistribution}) to ebits. We also allow Alice and Bob to decrease their net entanglement consumption by generating ebits. We first show how to obtain such a bound on non-interactive protocols, by adapting an argument from~\cite{LD15}. We get the following bound on the net entanglement consumption~$e$.

\begin{proposition}[Converse One-Shot Entanglement Consumption]
Let $\epsilon_1 \in (0, 1)$ and $\epsilon_2 \in (0, 1 - \epsilon_1)$ and $\rho \in \D_= (A \otimes  B \otimes  C)$ with purification register $R$. Then, the sum of the net entanglement consumption $e(\Pi)$ and quantum communication cost $q(\Pi)$ of every quantum state redistribution $\Pi$ of $\rho_{ABC}$ with error $\epsilon_1$ is lower bounded by
\begin{align}
e(\Pi) + q(\Pi) &\geq H_{\min}^{\epsilon_2} (BC)_\rho -  H_{\min}^{\epsilon_1 + \epsilon_2} (B)_{\rho}\,.
\end{align}
\end{proposition}

\begin{proof}
We consider a protocol with the same structure as in Proposition~\ref{prop:converse}, but we keep track of the environment registers $E_1$ and $E_2$ for the isometric extension of the encoding and decoding operations, respectively. That is, we consider the purified states
\begin{align}
\sigma^{T_A^{out} A^{\prime} Q E_1 B T_B^{in} R}\quad\mathrm{and}\quad\theta^{T_A^{out} T_B^{out} A^{\prime}  B^{\prime} C^{\prime} E_1 E_2 R}\,.
\end{align}
Take $\hat{\theta}$ such that $P (\hat{\theta}^{A R }, \rho^{AR} ) \leq \epsilon_2$ and $H_{\min}^{\epsilon_2} (A  R )_{\rho } =  H_{\min} (A  R )_{\hat{\theta}}$. Let $\phi_{out}^{T_A^{out} T_B^{out}}$ be a maximally entangled state and $\phi_E^{E_1 E_2}$ be a normalized pure state such that
\begin{align}
P (\rho^{ABCR} \otimes \phi_{out}^{T_A^{out} T_B^{out}} \otimes \phi_E^{E_1 E_2}, \theta^{ABCR T_A^{out} T_B^{out} E_1 E_2}) = P (\rho^{ABCR} \otimes \phi_{out}^{T_A^{out} T_B^{out}}, \theta^{ABCR T_A^{out} T_B^{out}}) \leq \epsilon_1\,.
\end{align}
Then we get the following chain of inequalities,
\begin{align}
H_{\min}^{\epsilon_2} (BC)_\rho + \log |T_A^{out}|  & = H_{\min}^{\epsilon_2} (AR)_\rho+ \log |T_A^{out}| \\
& \leq H_{\min} (AR)_{\hat{\theta}} + H_{\min} (T_A^{out})_{\phi_{out}} + H_{\min} (E_1)_{\phi_E} \\
& = H_{\min} (AR T_A^{out} E_1)_{\hat{\theta} \otimes \phi_{out} \otimes \phi_E} \\
& \leq H_{\min}^{\epsilon_1 + \epsilon_2} (AR T_A^{out} E_1)_{\theta} \\
& = H_{\min}^{\epsilon_1 + \epsilon_2} (AR T_A^{out} E_1)_{\sigma} \\
& = H_{\min}^{\epsilon_1 + \epsilon_2} (QB T_{B_{in}})_{\sigma} \\
& \leq H_{\min}^{\epsilon_1 + \epsilon_2} (B)_{\sigma} + \log |Q| + \log |T_B^{in}| \\
& = H_{\min}^{\epsilon_1 + \epsilon_2} (B)_{\rho} + \log |Q| + \log |T_B^{in}|\,.
\end{align}
The first equality follows since $\rho^{ABCR}$ is pure. The first inequality is by the choice of $\hat{\theta}$, because $\phi_{out}^{T_A^{out}}$ is a maximally mixed state and because $H_{\min} (E_1)_\phi \geq 0$. The second equality is because the min-entropy of product states is additive. The second inequality is by the choice of $\phi_E$ and by the triangle inequality and monotonicity of the purified distance.
The third equality is because the decoding operation does not affect the registers $A R T_A^{out} E_1$, and the fourth is because the state $\sigma^{T_A^{out} A^{\prime} Q E_1 B T_B^{in} R}$ is pure. The last inequality is by Lemma~\ref{lem:dimbound}, and the last equality because the encoding operation leaves the $B$ register untouched.
\end{proof}

This leads to the following strong converse on the entanglement consumption. The proof follows from a similar application of the AEP that was used in the strong converse for quantum communication. We do not repeat the details.

\begin{theorem}[Strong Converse Entanglement Consumption]
For all $\rho \in \D_= (A \otimes  B \otimes  C)$  purified by $\rho^{ABCR}$ for some register $R$ and $\mu > 0$, there exists $c > 0, n_0 \in \mathbb{N}$ such that for all $n \geq n_0$ and all quantum state redistribution $\Pi_n $ of $(\rho^{ABC})^{\otimes n}$ with error $1 - 2^{- c n}$, the sum of the net entanglement consumption $e(\Pi_n)$ and  quantum communication cost  $q (\Pi_n)$ of $\Pi_n$ satisfies
\begin{align}
e(\Pi_n) + q (\Pi_n) \geq n\cdot\Big[H(C|B)-\mu\Big]\,.
\end{align}
\end{theorem}


\subsection{Total Resource Consumption}\label{sec:conventint}

We are now ready to prove a lower bound on the total amount of resource consumption even if we allow for interaction between Alice and Bob. Notice that since the feedback from Bob to Alice can be used to distribute entanglement, it must be accounted for in the total resource consumption. We allow Alice and Bob to decrease their entanglement consumption by generating ebits. We consider the same model for communication as in Section~\ref{sec:back-comm}.

\begin{proposition}[Converse One-Shot Total Resources]\label{prop:multtotal}
Let $\epsilon_1 \in (0, 1)$ and $\epsilon_2 \in (0, 1 - \epsilon_1)$ and $\rho \in \D_= (A \otimes  B \otimes  C)$ with purification register $R$. Then, for every $M$-message quantum state redistribution $\Pi$ of $\rho$ with error $\epsilon_1$, the sum of the net entanglement consumption $e(\Pi)$ with the total quantum communication cost $QCC_{A \rightarrow B} (\Pi) + QCC_{B \rightarrow A}$  of $\Pi$ is lower bounded by
\begin{align}
e(\Pi) + QCC_{A \rightarrow B}(\Pi) + QCC_{B \rightarrow A} (\Pi) &\geq H_{\min}^{\epsilon_2} (BC)_\rho -  H_{\min}^{\epsilon_1 + \epsilon_2} (B)_{\rho}\,.
\end{align}
\end{proposition}

Note that there is no dependence on the number $M$ of messages in these lower bounds. Hence, the strong converse for quantum state redistribution proved from the corresponding bounds in the preceding section also hold if we account for feedback.

\begin{theorem}[Strong Converse Total Resources]
For all $\rho \in \D_= (A \otimes  B \otimes  C)$  purified by $\rho^{ABCR}$ for some register $R$ and $\mu > 0$, there exists $c > 0, n_0 \in \mathbb{N}$ such that for all $n \geq n_0$ and all $M$, every $M$-message
quantum state redistribution $\Pi_n^M $ of $(\rho^{ABC})^{\otimes n}$ with error $1 - 2^{- c n}$, the sum of the net entanglement consumption $e(\Pi_n)$ with the total quantum communication cost  $QCC_{A \rightarrow B} (\Pi_n) + QCC_{B \rightarrow A} (\Pi_n)$ of $\Pi_n$ satisfies
\begin{align}
e(\Pi_n) + QCC_{A \rightarrow B} (\Pi_n) + QCC_{B \rightarrow A} (\Pi_n) \geq n\cdot\Big[H(C|B)-\mu\Big]\,.
\end{align}
\end{theorem}

\begin{proof}[Proof of Proposition~\ref{prop:multtotal}]
The proof is  similar to the one in the single round case. We consider a $M$-message protocol with local isometric processing and arbitrary pre-shared entanglement $\phi_{in}^{T_A^{in} T_B^{in}}$, for the case of even $M$ (the case of odd $M$ follows similarly). The $A_i$ registers are Alice's quantum memory registers for round $i$, for odd $i$ and after application of $U_i$, similarly for $B_i$ on Bob's side and even $i$, and the $C_i$ registers are the communication registers exchanged by Alice and Bob in round $i$, from Alice to Bob for odd $i$, and from Bob to Alice for even $i$. We use the following notation,
\begin{align}
\rho_0 := \rho^{ABCR} \otimes \phi^{T_A T_B},\;\rho_1 := U_1 (\rho_0), \rho_2 := U_2 (\rho_1),\;\cdots,\;\rho_{M+1} := U_{M+1} (\rho_M)\,.
\end{align}
For $\phi_{out}$ a maximally entangled state it holds that
\begin{align}
P (\rho_{M+1}^{ABCR T_A^{out} T_B^{out}}, \rho^{ABCR} \otimes \phi_{out}^{T_A^{out} T_B^{out}}) \leq \epsilon_1\,.
\end{align}
Now take $\hat{\theta}$ such that $P (\hat{\theta}^{A R }, \rho^{AR} ) \leq \epsilon_2$ as well as $H_{\min}^{\epsilon_2} (A  R )_{\rho } =  H_{\min} (A  R )_{\hat{\theta}}$, let $\phi_{out}^{T_A^{out} T_B^{out}}$ be a maximally entangled state, and let $\phi_E^{A^\prime B^\prime}$ be a normalized pure state such that
\begin{align}
P (\rho^{ABCR} \otimes \phi_{out}^{T_A^{out} T_B^{out}} \otimes \phi_E^{A^\prime B^\prime}, \rho_{M+1}^{ABCR T_A^{out} T_B^{out} A^\prime B^\prime}) = P (\rho^{ABCR} \otimes \phi_{out}^{T_A^{out} T_B^{out}}, \theta^{ABCR T_A^{out} T_B^{out}}) \leq \epsilon_1\,.
\end{align}
Then we get the following chain of inequalities,
\begin{align}
H_{\min}^{\epsilon_2} (BC)_\rho + \log |T_A^{out}|  & = H_{\min}^{\epsilon_2} (AR)_\rho+ \log |T_A^{out}| \\
& = H_{\min} (AR)_{\hat{\theta}} + H_{\min} (T_A^{out})_{\phi_{out}} \\
& \leq H_{\min} (AR)_{\hat{\theta}} + H_{\min} (T_A^{out})_{\phi_{out}} + H_{\min} (A^\prime)_{\phi_E} \\
& = H_{\min} (AR T_A^{out} A^\prime)_{\hat{\theta} \otimes \phi_{out} \otimes \phi_E} \\
& \leq H_{\min}^{\epsilon_1 + \epsilon_2} (AR T_A^{out} A^\prime)_{\sigma_{M+1}} \\
& = H_{\min}^{\epsilon_1 + \epsilon_2} (A_{M-1} C_M R)_{\sigma_{M}} \\
& \leq H_{\min}^{\epsilon_1 + \epsilon_2} (A_{M-1} R)_{\sigma_{M}} + \log |C_M| \\
& = H_{\min}^{\epsilon_1 + \epsilon_2} (B_{M} C_M)_{\sigma_{M}} + \log |C_M| \\
& = H_{\min}^{\epsilon_1 + \epsilon_2} (C_{M-1} B_{M-2})_{\sigma_{M-1}} + \log |C_M| \\
& \leq H_{\min}^{\epsilon_1 + \epsilon_2} (B_{M-2})_{\sigma_{M-1}} + \log |C_{M-1}| + \log |C_M| \\
& = H_{\min}^{\epsilon_1 + \epsilon_2} (A_{M-1} C_{M-1} R)_{\sigma_{M-1}} + \log |C_{M-1}| + \log |C_M| \\
& = H_{\min}^{\epsilon_1 + \epsilon_2} (A_{M-3} C_{M-2} R)_{\sigma_{M-2}} + \log |C_{M-1}| + \log |C_M| \\
& \leq H_{\min}^{\epsilon_1 + \epsilon_2} (A_{M-3}  R)_{\sigma_{M-2}} + \log|C_{M-2}|+ \log |C_{M-1}| + \log |C_M| \\
& \leq \cdots \\
& \leq H_{\min}^{\epsilon_1 + \epsilon_2} (A_1  R)_{\sigma_{2}} + \log|C_{2}|+ \cdots + \log |C_M| \\
& = H_{\min}^{\epsilon_1 + \epsilon_2} (B_2  C_2)_{\sigma_{2}} + \log|C_{2}|+ \cdots + \log |C_M| \\
& = H_{\min}^{\epsilon_1 + \epsilon_2} (C_1 B T_B^{in})_{\sigma_{1}} + \log|C_{2}|+ \cdots + \log |C_M| \\
& \leq H_{\min}^{\epsilon_1 + \epsilon_2} ( B )_{\sigma_{1}} + \log|T_B^{in}| + QCC_{A \rightarrow B} (\Pi) + QCC_{B \rightarrow A} (\Pi)  \\
& = H_{\min}^{\epsilon_1 + \epsilon_2} ( B )_{\rho} + \log|T_B^{in}| + QCC_{A \rightarrow B} (\Pi) + QCC_{B \rightarrow A} (\Pi)\,.
\end{align}
The first equality follows since $\rho^{ABCR}$ is pure, and the second by the choice of $\hat{\theta}$ and because $\phi_{out}^{T_A^{out}}$ is a maximally mixed state. The first inequality follows because $H_{\min} (A^\prime)_{\phi_E} \geq 0$. The third equality is because the min-entropy of product states is additive. The second inequality is by the choice of $\phi_E$ and by the triangle inequality and monotonicity of the purified distance. The fourth equality is by isometric invariance of the min-entropy, and the third inequality is by Lemma~\ref{lem:dimbound}.
The fifth equality is because $\sigma_M^{A_{M-1 } R B_M C_M}$ is a pure state, and the sixth is by isometric invariance. The fourth inequality is again by Lemma~\ref{lem:dimbound}. The next two equalities are because $\sigma_{M-1}^{B_{M-2} A_{M-1} C_{M-1} R}$ is a pure state and by isometric invariance, respectively. The next inequality is by Lemma~\ref{lem:dimbound}, and the following chain of inequalities is by repeateadly applying the last few ones. The two equalities leading to the last inequality are because $\sigma_2^{A_1 R B_2 C_2}$ is a pure state and by isometric invariance, respectively, while the last inequality is by two applications of Lemma~\ref{lem:dimbound}, and also by definition of $QCC_{A \rightarrow B} (\Pi)$ and $QCC_{B \rightarrow A} (\Pi)$. Finally, the last equality is because $U_1$ leaves the $B$ register untouched.
\end{proof}


\section{Conclusion}\label{sec:conclusion}

We have proved that one-shot quantum state redistribution of $\rho^{ABCR}$ up to error $\epsilon$ can be achieved at communication cost at most
\begin{align}
\frac{1}{2}\Big[H_{\max}^{\epsilon} (C|B)_\rho - H_{\min}^{\epsilon} (C | BR)_\rho\Big] + O\big(\log (1/\epsilon)\big)\,.
\end{align}
when free entanglement assistance is available (independently, this bound has also been derived in~\cite{DHO11}). The structure of the protocol achieving this performs a decomposition of state redistribution into two state merging protocols. Such a decomposition was proposed in~\cite{Opp08} in order to achieve asymptotically tight rates. Note that we could alternatively use a decomposition into a state merging and a state splitting protocol, as proposed in~\cite{YBW08}, to achieve similar bounds. An important technical ingredient for our proof is the bi-decoupling lemma that we prove as an extension of the well-known decoupling theorem~\cite{BCR11}. A similar lemma was derived in~\cite{YBW08}, with bounds in terms of dimensions rather than conditional min-entropies. This lemma states that for two states on the same system $C$, there exists at least one unitary on $C$ that acts as a decoupling unitary for both states simultaneously, when parameters are appropriately chosen. Perhaps surprisingly, this idea allows us to smooth both the conditional min- and max- entropy terms appearing in our bounds, notwithstanding the fact that it is in general unknown how to simultaneously smooth marginals of overlapping quantum systems (see, e.g.,~\cite{DF13} and references therein).

We emphasize again that our achievability bound~\eqref{eq:achievability} has already found applications. In particular, one of the authors obtained the first multi-round direct sum theorem in quantum communication complexity~\cite{Tou14a}. However, it is known from the work on one-shot state merging and splitting~\cite{BCR11} that, for arbitrary shared entanglement, the bound~\eqref{eq:achievability} can in general not be optimal, and in fact for some states the achievable communication can be substantially lower. An interesting open problem is to obtain a tight characterization of the minimal quantum communication cost. Recent works on the R\'enyi generalizations of conditional mutual information in the quantum regime~\cite{BSW14} might enable to shed some light on this question. In particular, it would be of interest to link some version of our improved bound~\eqref{eq:qrst_idea} to a smooth version of the conditional max-information,
\begin{align}\label{eq:maxinfo_cond}
I_{\max}(C;R|B)_{\rho}:=D_{\max}\Big(\rho^{CBR}\big\|\big(\rho^{BR}\big)^{1/2}\big(\rho^{B}\big)^{-1/2}\rho^{BC}\big(\rho^{B}\big)^{-1/2}\big(\rho^{BR}\big)^{1/2}\Big)\,.
\end{align}
In turn this would also shine some light on the R\'enyi generalizations of the conditional mutual information in~\cite{BSW14}.

Finally, we have shown that our one-shot converse bounds imply a strong converse for quantum state redistribution in the iid asymptotic limit (that even holds when allowing for feedback).


\paragraph{Acknowledgments}

We thank Felix Leditzky for discussions about Ref.~\cite{LD15}. These discussions were the starting point for the derivations in Sections~\ref{sec:convent} and~\ref{sec:conventint}. We acknowledge discussions with Renato Renner, Mark Wilde, and J\"urg Wullschleger. The hospitality of the Banff International Research Station (BIRS) during the workshop “Beyond IID in Information Theory” (5-10 July 2016) is gratefully acknowledged (the work in Sections~\ref{sec:convent} and~\ref{sec:conventint} was performed there). MC was supported by a Sapere Aude grant of the Danish Council for Independent Research, an ERC Starting Grant, the CHIST-ERA project ``CQC'', an SNSF Professorship, the Swiss NCCR ``QSIT'' and the Swiss SBFI in relation to COST action MP1006. Most of this work was done while DT was a PhD student at Universit\'e de Montr\'eal, and was supported in part by a FRQNT B2 Doctoral research scholarship and by CryptoWorks21.


\appendix

\section{Miscellaneous Lemmas}

The following bound holds on the smooth max information~\cite[Lemma B.9]{BCR11}.

\begin{lemma}\label{lem:b9}
Let $\epsilon \geq 0$ and $\rho^{ABC} \in \D_= (A \otimes B \otimes C)$. Then, we have
\begin{align}
I_{\max}^\epsilon (A : BC)_\rho \leq I_{\max}^\epsilon (A: B)_\rho + 2 \log |C|\,.
\end{align}
\end{lemma}

We also need the same type of bound for smooth conditional min-entropy.

\begin{lemma}\label{lem:unlock}
Let $\rho^{ABC}\in\D_{=}(A\otimes B\otimes C)$ and $\epsilon\geq0$. Then, we have
\begin{align}
H_{\min}^{\epsilon}(A|B)_{\rho}\leq H_{\min}^{\epsilon}(A|BC)_{\rho}+2\log|C|\,.
\end{align}
\end{lemma}

\begin{proof}
Let $\tilde{\rho}^{AB}\in\B^{\epsilon}\big(\rho^{AB}\big)$ and $\sigma^{B}\in\D_{=}(B)$ such that $H_{\min}^{\epsilon}(A|B)_{\rho}=-D_{\max}(\tilde{\rho}^{AB}\|I^{A}\otimes\sigma^{B})=-\log\lambda$. We have
\begin{align}
\lambda\cdot I^{A}\otimes\sigma^{B}\geq\tilde{\rho}^{AB}\quad\Rightarrow\quad\lambda\cdot I^{A}\otimes\sigma^{B}\otimes\frac{I^{C}}{|C|}\geq\tilde{\rho}^{AB}\otimes\frac{I^{C}}{|C|}\,.
\end{align}
Now take $\tilde{\rho}^{ABC}\in\B^{\epsilon}\big(\rho^{ABC}\big)$ and with~\cite[Lemma B.6]{BCR11}, $|C|\cdot\tilde{\rho}^{AB}\otimes I^C\geq\tilde{\rho}^{ABC}$ we get
\begin{align}
\lambda|C|^{2}\cdot I^{A}\otimes\sigma^{B}\otimes\frac{I^{C}}{|C|}\geq\tilde{\rho}^{ABC}\,.
\end{align}
Hence, we can conclude the claim
\begin{align}
H_{\min}^{\epsilon}(A|B)_{\rho}=-\log\lambda&\leq-D_{\max}(\tilde{\rho}^{ABC}\|I^{A}\otimes\sigma^{B}\otimes\frac{I^{C}}{|C|})+2\log|C|\\
&\leq\sup_{\hat{\rho}^{ABC}\in\B^{\epsilon}(\rho^{ABC})}\sup_{\omega^{BC}\in\D_{=}(B\otimes C)}
D_{\max}(\hat{\rho}^{ABC}\|I_{A}\otimes\omega^{BC})+2\log|C|\\
&=H_{\min}^{\epsilon}(A|BC)_{\rho}+2\log|C|\,.
\end{align}
\end{proof}

Note that by duality, a similar result holds for smooth conditional max-entropy. Another bound on smooth conditional min-entropy we use is the following, which can be seen to follow from~\cite[Lemma 5]{RR12} and the data processing inequality.

\begin{lemma}\label{lem:dimbound}
Let $\rho^{ABC}\in\D_{=}(A\otimes B\otimes C)$ and $\epsilon\geq0$. Then, we have
\begin{align}
H_{\min}^{\epsilon}(AB|C)_{\rho}\leq H_{\min}^{\epsilon}(A|C)_{\rho}+\log|B|\,.
\end{align}
\end{lemma}

We also make use of the following variant of Uhlmann's theorem.

\begin{lemma}\label{lem:ull}
Let $\rho_1, \rho_2 \in \D_\leq (A)$ have purifications $\rho_1^{A R_1}, \rho_2^{A R_2}$. Then, there exists a partial isometry $V^{R_1 \rightarrow R_2}$ such that
\begin{align}
P\big(\rho_1^{A}, \rho_2^{A}\big) = P\Big(V \big(\rho_1^{A R_1}\big), \rho_2^{A R_2} \Big)\,.
\end{align}
\end{lemma}



\end{document}